\documentclass[journal]{IEEEtran}

\usepackage{cite}
\usepackage{amsmath}
\usepackage{array}
\usepackage{url}
\usepackage{siunitx}
\usepackage{color}
\usepackage{graphicx}
\usepackage{amssymb}
\usepackage{epstopdf}
\usepackage{hyperref}
\usepackage{footnote}
\usepackage{epsfig,psfrag}
\usepackage{color}
\usepackage{amsmath}
\usepackage{mathrsfs}  
\usepackage{bm}
\usepackage{multirow}
\usepackage[mode=buildnew]{standalone}
\usepackage{algorithm}
\usepackage{algorithmicx}
\usepackage{algpseudocode}
\usepackage{multicol}
\usepackage{amsthm}
\usepackage{cleveref}
\usepackage{mleftright}
\usepackage{verbatim}
\usepackage{enumitem}
\usepackage{stmaryrd}
\usepackage{float}
\usepackage{makecell}

\newtheorem{mylemma}{Lemma}

\newtheorem{myproposition}{Proposition}
\newtheorem{mydefinition}{Definition}
\newtheorem{myremark}{Remark}

\newtheorem{myexample}{Example}

\DeclareMathOperator*{\argmax}{arg\,max}
\DeclareMathOperator*{\argmin}{arg\,min}

\ifCLASSINFOpdf
\else
\fi
\ifCLASSOPTIONcompsoc
 \usepackage[caption=false,font=normalsize,labelfont=sf,textfont=sf]{subfig}
\else
 \usepackage[caption=false,font=footnotesize]{subfig}
\fi

\begin{document}

\title{Verification of Neural Network Control Systems using Symbolic Zonotopes and Polynotopes\\ \color{red}\normalsize\rule{0pt}{14pt}This work has been submitted to the IEEE for possible publication.\\Copyright may be transferred without notice, after which this version may no longer be accessible.\color{black}}

\author{Carlos Trapiello, Christophe Combastel, and Ali Zolghadri, \IEEEmembership{Senior Member, IEEE}
\thanks{The authors are with Univ. Bordeaux, CNRS, Bordeaux INP, IMS, UMR 5218, F-33400 Talence, France; C. Trapiello is also with the Supervision, Safety and Automatic Control Research Center (CS2AC), of the Universitat Politècnica de Catalunya (UPC), 08222 Terrassa, Spain.}}

\markboth{IEEE Transactions on ---, VOL. XX, NO. XX, XXXX 2022}{Trapiello \MakeLowercase{\textit{et al.}}: Toward Verification of Neural Network Control Systems using Symbolic Zonotopes and Polynotopes}

\maketitle

\begin{abstract}
Verification and safety assessment of neural network controlled systems (NNCSs) is an emerging challenge. To provide guarantees, verification tools must efficiently capture the interplay between the neural network and the physical system within the control loop. In this paper, a compositional approach focused on inclusion preserving long term symbolic dependency modeling is proposed for the analysis of NNCSs. First of all, the matrix structure of symbolic zonotopes is exploited to efficiently abstract the input/output mapping of the loop elements through (inclusion preserving) affine symbolic expressions, thus maintaining linear dependencies between interacting blocks. Then, two further extensions are studied. Firstly, symbolic polynotopes are used to abstract the loop elements behaviour by means of  polynomial symbolic expressions and dependencies. Secondly, an original input partitioning algorithm takes advantage of symbol preservation to assess the sensitivity of the computed approximation to some input directions. The approach is evaluated via different numerical examples and benchmarks. A good trade-off between low conservatism and computational efficiency is obtained.
\end{abstract}


\IEEEpeerreviewmaketitle


\section{Introduction} \label{sec:introduction}

\IEEEPARstart{T}{he} proliferation of data and access to ever-increasing computational power have fueled a renewed interest in deep neural-networks (NNs). These networks have shown a significant ability to address classification/estimation/control tasks that can hardly be formalized and designed from knowledge-based models. However, despite their impressive ability for solving complex problems, it is well known that NNs can be vulnerable to small perturbations or adversarial attacks~\cite{szegedy2013intriguing,kurakin2018adversarial}. This lack of robustness (or fragility) represents a major barrier for their application to safety-critical system where safety assurances are of primary importance. For example, in Guidance, Navigation and Control of flight systems, one must ensure that some output/state trajectories remain inside a flight envelope when some inputs explore a given region. The above issues have fostered a large amount of works that analyze the sensitivity  to  local disturbances of NNs in isolation (open-loop), as well as the satisfaction of pre-/post- safety conditions~\cite{liu2021algorithms}.  Nevertheless, as reported in~\cite{johnson2021arch}, reasoning about the safety verification of neural-network control systems (NNCSs), where the NN is used as a feedback controller, still remains a key challenge that requires tractable methods  capable of efficiently integrating the heterogeneous components that make up the control loop.

This paper focuses on the reachability analysis of NNCSs which, in turn, allows for formal reasoning about the satisfaction of safety properties  (reachability of a target set, or the avoidance of non-secure sets of states). A key challenge in NNCSs reachability analysis is to successfully retain the system-controller interplay by preserving (at each time instant) the dependencies between relevant variables. This, in fact,  discourages a direct application of off-the-shelf verification tools which, although able to compute accurate output bounds for elements in isolation, return coarse approximations when iteratively concatenated for the analysis of closed-loop systems since most of (if not all) the I/O dependencies are quickly broken/lost during the computations~\cite{yang2019efficient,tran2020nnv,claviere2021safety}. Furthermore, effective NNCSs verification tools must be able to assess the system state during (relatively) large time intervals. The above issues motivate the development of computationally efficient  analysis methods capable of  capturing the interaction between the control loop elements  while granting a good scalability both in the system dimensions and in the time horizon length.

Another relevant factor that should be taken into account is the size of the initial state set under study. Mainly, the performance of open- and closed-loop verification techniques that are based on (locally) abstracting the system  non-linearities, deteriorates considerably for large initial sets. A common approach to address this issue, particularly in NNCSs verification problems where the number of dimensions is relatively small,  is to recast the initial reachability problem into simpler subproblems that analyze a subset of the initial conditions~\cite{yang2019efficient,schilling2021verification,claviere2021safety,ivanov2020verifying,sidrane2022overt,everett2021efficient,xiang2020reachable,everett2020robustness}. Nonetheless, the design of efficient and scalable partitioning strategies, specially in closed-loop verification schemes,  remains also an open problem.


\paragraph*{Related work} Preserving dependencies for NNCS verification has spurred on some recent studies. In~\cite{dutta2019reachability}, the authors abstract the I/O mapping of a ReLU NN controller using a polynomial expression (plus an error interval). The polynomial rule is obtained by regression of I/O samples, whereas a sound error term is derived from solving a mixed-integer program (MIP). In a similar fashion, \cite{huang2019reachnn} uses Bernstein polynomials to abstract the NN controller. A theoretical and a sampled-based method is proposed to compute the error term based on the  Lipschitz constant of the NN. Although both approaches preserve the system-controller interplay, they are computationally expensive, scaling poorly with the number of NN inputs while requiring to be iteratively repeated  for each output. In~\cite{ivanov2020verifying}  a NN with differentiable activation functions is transformed into an equivalent hybrid system built upon Taylor models that  retain dependencies. However, this approach is not applicable to ReLU functions and the number of states (resp. modes) of the hybrid automaton scales with the number of neurons (resp. layers).

Other approaches preserve system-controller dependencies by formulating the reachable set computation as an optimization problem. The work~\cite{hu2020reach} proposes a semidefinite program for reachability analysis based on the abstraction of the NN non-linearities using quadratic constraints~\cite{fazlyab2020safety}. \cite{everett2021efficient} relies on the tool  CROWN~\cite{zhang2018efficient}, and preserves system-controller interaction by solving LP-programs. However,  dependencies are broken from one sample to the next. In~\cite{sidrane2022overt}, the closed-loop is firstly abstracted  as a conjunction of piecewise-linear functions, and then analyzed using ReLU NNs verification tools like~\cite{tjeng2017evaluating,katz2017reluplex}.

On the other hand, other works address the reachability problem by chaining different verification tools. In~\cite{yang2019efficient}, the authors combine a polytopic abstraction of the dynamical system with the tool Sherlock~\cite{dutta2017output}, that is used to bound the NN controller outputs. The tool NNV~\cite{tran2020nnv}, integrates the non-linear dynamics  reachability tool CORA~\cite{althoff2015introduction} with a star sets abstraction of ReLU NN controllers~\cite{tran2020verification}. Besides,  \cite{claviere2021safety} combines validated simulation to soundly approximate the dynamical system with common tools for NN output bounding like DeepPoly~\cite{singh2019abstract}. In all the above works, dependencies are broken in the switch  between the different tools. This latter issue is somehow palliated in~\cite{schilling2021verification}, where the authors use second order zonotopes (i.e. zonotopes with generators matrix size $n \times 2n$) as an interface between system and NN controller analysis tools. Although capable to retain first order dependencies in the system-to-controller (and controller-to-system) set transformations, dependencies in the I/O mapping of the NN controller are broken. 

Focusing on partitioning strategies, in \cite{wang2018formal} the gradient of a ReLU NN (open-loop) is used to decide the next input direction to be bisected, whereas in \cite{xiang2018output} a uniform grid of the initial set is employed. Other works propose a simulation-based splitting strategy. In \cite{xiang2020reachable}, the bisection is guided by comparing  the  interval bound of Monte-Carlo samples with a guaranteed  Interval Bound Propagation \cite{gowal2018effectiveness} of the initial subsets. Working in a similar fashion, \cite{everett2020robustness} proposes a simulation-guided framework that unifies standard NN output bounding tools. The decision on the bisection order is based on the distance to the simulation samples enclosure. A closed-loop implementation of the latter algorithm is reported in~\cite{everett2021efficient}.


\paragraph*{Contributions} This paper takes a new and original direction based on symbolic zonotopes (s-zonotopes) as a generic tool for the closed-loop verification of discrete-time NN controlled systems. The generators (matrix) representation of s-zonotopes enables to efficiently abstract the input-output mappings of the NN controller and non-linear physical system through (inclusion preserving) affine symbolic expressions. The evolution of the closed-loop system can then be bounded in a propagation-based fashion that benefits from the efficient computation of basic operations granted by s-zonotopes, while preserving system-controller linear dependencies. Besides,  the computational complexity of the verification tool can be fixed by limiting (reducing) the number of independent symbols. Simulations show the good performance/computational efficiency trade-off granted by this approach.

Furthermore, two extensions are proposed. On the one hand, the use of polynomial symbolic expressions to abstract the input-output mapping of the loop elements is explored. In particular, symbolic polynotope (s-polynotope) structures \cite{combastel2020functional} are used to enclose the NN activation functions graph via the non-convex sets that arise from the polynomial map of interval symbols. Polynomial abstractions enable to reduce the conservatism induced by linear relaxations, at the price of increasing the computation needs.

On the other hand, the  symbols preservation throughout the control loop is exploited to develop a  smart partitioning strategy of the initial conditions set. The proposed algorithm reasons upon the influence of the input symbols in the output set in order to select which dimension to bisect next, and upon the influence of the (independent) error symbols to assess the quality of each over-approximation. 


\paragraph*{Structure} The paper is organized as follows. \Cref{sec:preliminaries} is devoted to some useful preliminaries. \Cref{sec:statement} introduces the problem statement. Then,  \Cref{sec:verification} analyzes the closed-loop verification using s-zonotopes. In  \Cref{sec:polynotope} the use of s-polynotopes is investigated, whereas the input partitioning algorithm is detailed in \Cref{sec:partitioning}. \Cref{sec:simulations} presents simulation results. Finally,  some concluding remarks are provided in \Cref{sec:conclusions}. 


\paragraph*{Notation} The following notations are used along this work. $\mathbb{R}^n$, $\mathbb{R}^{m\times n}$ and $\mathbb{N}$ denote the $n$ dimension Euclidean space, the $m \times n$ dimensional Euclidean space and the set of non-negative integers, respectively. The notation $v_i$ stands for the $i$-th element of vector $v$ and $M_{[i,:]}$ ($M_{[:,j]}$) for the $i$-th row ($j$-th column) of matrix $M$. The 1-norm  of the (row) vector $v$ is $\|v\|_1 = |v|\mathbf{1}$, with $|.|$  the elementwise absolute value, and $\mathbf{1}$ a column vector of ones of appropriate size. $diag(v)$ returns a square diagonal matrix with the elements of vector $v$ in the main diagonal, whereas $card(\cdot)$ gives the cardinal.


\section{Symbolic dependencies in set computations} \label{sec:preliminaries}

This section provides preliminary concepts which will be used in the following sections. Throughout this article, $s$ refers to an indexed family of distinct symbolic variables of type unit interval, that is, $\forall i \in \mathbb{N}$, the symbol $s_i$ (uniquely identified by the integer $i$) refers to a scalar real variable the value of which is only known to belong to the unit interval $\mathcal{D}(s_i)=[-1,+1] \subset \mathbb{R}$. Also, $\mathcal{D}(s)=[-1,+1]^{card(s)}$. In other words, the a priori unknown value $\iota s_i$ taken by the symbolic variable $s_i$ satisfies $\iota s_i \in \mathcal{D}(s_i)$. The generic notation $\iota$ which reads as "interpretation/valuation of" helps disambiguate between symbols (syntax) and values (semantics) \cite{combastel2020functional}. Note that, in general, several interpretations may coexist. Set-valued interpretations take sets as values. In the following, consistently with the definition domain $\mathcal{D}(s_i)$ related to $s_i$, the set-valued interpretation of each symbolic variable $s_i$ will be ${s_i}_{|\iota} = [-1,+1]$. In addition,  the integer-valued vector $I$ is used to uniquely identify a set of symbols, for example,  vector $I = [1, \,5, \, 3]$ identifies the symbols $s_1$, $s_5$ and $s_3$. For brevity of notation, $s_I$ denotes the column vector $[s_i]_{i\in I}$.

\begin{mydefinition}[s-zonotope \cite{combastel2020functional}] \label{def:s-zonotope}
A symbolic zonotope $\mathcal{X}_{|s}$ is an affine symbolic function that can be written in the form $c + Rs_I$ where vector $c$ and matrix $R$ do not depend on the symbolic variables in $s_I$. Notation: $\mathcal{X}_{|s} = \langle c,R,I\rangle _{|s} = c + R s_I$.
\end{mydefinition}

\begin{mydefinition}[e-zonotope \cite{combastel2020functional}] \label{def:s-zonotope}
The e-zonotope $\mathcal{X}_{|\iota}$ related to the s-zonotope $\mathcal{X}_{|s} = \langle c,R,I\rangle_{|s} = c + R s_I$ is the set-valued interpretation of $\mathcal{X}_{|s}$ as $\mathcal{X}_{|\iota} = \langle c,R,I\rangle _{|\iota} = \{c + R \sigma | \sigma \in \mathcal{D}(s_I)\}$. 
\end{mydefinition}

A basic example is : given $i \in \mathbb{N}$, $\langle 0,1,i\rangle _{|s} = s_i$ (symbolic expression corresponding to the $i$-th symbol in $s$) and $\langle 0,1,i\rangle _{|\iota} = \mathcal{D}(s_i) = [-1,+1]$ (set-valued interpretation of $s_i$). More generally, s-zonotopes  and their interpretation as e-zonotopes make it possible to explicitly perform operations either at symbolic/syntactic level ($._{|s}$) or at semantic level ($._{|\iota}$). 

\begin{myremark}
In this work, all symbols being of type unit interval, $c$ being a real vector and $R$ a real matrix, $\mathcal{X}_{|\iota}$ is a classical zonotope $\langle c,R \rangle$ with center $c$ and generator matrix $R$ (for extensions to other symbol types, see \cite{combastel2020functional}). Note that $\langle 0,1,i\rangle_{|s}-\langle 0,1,j\rangle_{|s} = s_i - s_j = 0$ for $i=j$, whereas $\langle 0,1,i\rangle_{|\iota}-\langle 0,1,j\rangle_{|\iota} = [-2,+2]$ for all $(i,j)$, that is, even for $i=j$. Operating at the symbolic/syntactic level thus permits more accurate set evaluations by preserving trace of symbolic dependencies. This is a key point to prevent from pessimistic outer approximations induced by the so-called dependency problem \cite{moore2009introduction} affecting natural interval arithmetic and other classical set-based operations only considering the semantic level.
\end{myremark}

From a computational point of view, an s-zonotope is defined by storing the triplet $(c,R,I)$. Due to their affine structure, a key aspect is to efficiently trace  the identifier $i \in I$ of the symbol that multiplies each column of the matrix $R$.  To that end, Matrices with Labeled Columns (MLCs), constitute a data structure featuring columnwise sparsity: It is defined by the pair $(R,I)$ that allows for efficiently recasting standard operations involving s-zonotopes as set-operations on the identifiers vector ($I$) and column-wise operations in the projection matrices ($R$). For how to translate  operations such as sum or linear image onto a computational platform using MLCs the interested reader can refer to \cite{combastel2020distributed}.

Due to their relevance in further developments, the following operations involving s-zonotopes are briefly recalled.

\begin{mylemma}[common symbols \cite{combastel2020distributed}]\label{lem:common_id}
Any two s-zonotopes $\mathcal{X}_{|s} = \langle c_x, R,I\rangle_{|s}$ and $\mathcal{Y}_{|s} = \langle c_y, G, J\rangle_{|s} $, can be rewritten using a common set of symbols $s_K$ as $\mathcal{X}_{|s} = \langle c_x , \tilde{R}, K\rangle_{|s}$ and  $\mathcal{Y}_{|s} = \langle c_y, \tilde{G}, K\rangle_{|s}$, with 
\begin{equation}\label{eq:mergeid}
\begin{aligned}
      \tilde{R} &= \begin{bmatrix} R_1, & R_2, & 0 \end{bmatrix}, \quad 
      \tilde{G} = \begin{bmatrix} G_1, & 0, & G_2 \end{bmatrix}, \\
      K &= \begin{bmatrix}I\cap J; & I \setminus J; & J \setminus I \end{bmatrix}. 
\end{aligned}
\end{equation}
\end{mylemma}

Matrices $(R_1,G_1)$ in  \Cref{lem:common_id} may be empty matrices if $I\cap J$ is empty (similarly for $R_2$ and  $I \setminus J$ or $G_2$ and $J \setminus I$).

\begin{mydefinition}[basic operations \cite{combastel2020distributed}] \label{def:operations}
Given two s-zonotopes $\mathcal{X}_{|s}$ and $\mathcal{Y}_{|s}$ with a common set of symbols $s_K$ as in \Cref{lem:common_id}, then  their sum and vertical concatenation are the s-zonotopes 
\end{mydefinition}
\vspace*{-8mm}
\begin{align}
    & \mathcal{X}_{|s} + \mathcal{Y}_{|s} = \big\langle c_x+c_y,  [R_1+G_1, \, R_2, \, G_2], K \big\rangle_{|s}, \label{eq:sum} \\
    &[\mathcal{X}_{|s};\mathcal{Y}_{|s}] = \bigg \langle \begin{bmatrix}
c_x \\ c_y
\end{bmatrix}, \begin{bmatrix}
R_1 & R_2 & 0\\ G_1 & 0 & G_2
\end{bmatrix}, K\bigg \rangle_{|s}. \label{eq:vcat}
\end{align}

\begin{mydefinition}[inclusion \cite{combastel2020functional}] \label{def:inclusion}
The s-zonotope $\mathcal{Y}_{|s}$  is said to include the s-zonotope $\mathcal{X}_{|s}$, if the set-valued interpretation of $\mathcal{Y}_{|s}$ includes the set-valued interpretation of $\mathcal{X}_{|s}$. In other words, the expression $\mathcal{X}_{|s} \subset \mathcal{Y}_{|s}$ interprets as $\mathcal{X}_{|\iota} \subset \mathcal{Y}_{|\iota}$.
\end{mydefinition}

\Cref{def:inclusion}  paves the way for rewriting rules (at symbolic level) that may be either inclusion preserving or inclusion neutral or none of both at set-based evaluation (semantic) level. A more formal treatment of this topic can be found in the definition 27 (rewriting rules and inclusion) in \cite{combastel2020functional}, 
where a definition of inclusion functions is also given in definition 2. 

\begin{mydefinition}[reduction \cite{combastel2020functional}] \label{def:reduction}
The reduction operator $\downarrow_q $ transforms an s-zonotope $\mathcal{X}_{|s} = \langle c, R, I\rangle_{|s}$ into a new s-zonotope $\tilde{\mathcal{X}}_{|s} = \downarrow_q \mathcal{X}_{|s} = \langle c, G, J\rangle_{|s}$, such that  $\tilde{\mathcal{X}}_{|s}$ includes $\mathcal{X}_{|s}$ while depending on at most $q$ symbols, i.e. $card(J) \leq q$.
\end{mydefinition}

Reduction is thus an inclusion preserving transform. In \Cref{def:reduction}, $I \cap J \neq \emptyset $ is not mandatory but often useful to prevent from a further propagation of conservative approximations, while controlling the complexity of $\tilde{\mathcal{X}}_{|s}$ through the maximum number $q$ of its symbols/generators. In this context, preserving the more significant symbols/dependencies is often beneficial: as in \cite{combastel2020distributed}, 
if $p>q$ a common practice is to replace the $p-q+1$ less important symbols by a new independent one while guaranteeing the inclusion $\mathcal{X}_{|s} \subseteq \tilde{\mathcal{X}}_{|s}$. Besides, note that  new symbols introduced to characterize independent behaviors must be uniquely identified. Wherever needed, the generation of a vector of $n$ new  unique symbols identifiers is denoted as~$!(n)$. The generation of a pre-specified number of identifiers can be attained  by implementing, for example, the Unique Symbols Provider (USP) service introduced in~\cite{combastel2020distributed}.


\section{Problem statement} \label{sec:statement}


\subsection{System description} \label{sec:NNCSS}

Consider the  interconnection of a discrete-time non-linear dynamic model \eqref{eq:nonlinear_model} and a neural network. The physical system is modeled as:
\begin{equation}\label{eq:nonlinear_model}
    x(t+1) = f(x(t),u(t),w(t)),
\end{equation}
\noindent where $x(t) \in \mathbb{R}^{n_x}$ and $u(t) \in \mathbb{R}^{n_u}$ respectively refer to the state and the control input at time step $t \in \mathbb{N}$. For all $t\geq 0$,  vector $w(t)$   accounts for modeling errors and process disturbances and satisfies $w(t) \in \mathcal{W} = [-1, \, +1]^{n_w}$.

The system \eqref{eq:nonlinear_model} is controlled by a  state-feedback controller $g(x(t)) : \mathbb{R}^{n_x} \mapsto \mathbb{R}^{n_u}$ parameterized by an $l$-layer feed-forward fully connected neural network. The map $x \mapsto g(x)$ is described by the following recursive equations
\begin{equation} \label{eq:NN_model}
    \begin{aligned}
    &x^{(0)}  = x,\\
    &x^{(k+1)} = \phi^{(k)}(W^{(k)}x^{(k)} + b^{(k)}), \quad   k  = 0,...,l-1,\\
    & g(x) = W^{(l)}x^{(l)} + b^{(l)},
    \end{aligned}
\end{equation}
\noindent where  $x^{(k)}\in \mathbb{R}^{n_k} $ are the outputs (post-activation) of the $k$-th layer. The weight matrix $W^{(k)} \in \mathbb{R}^{n_{k+1}\times n_k}$ and bias $b^{(k)} \in \mathbb{R}^{n_{k+1}}$ define the affine mapping $z^{(k)} = W^{(k)}x^{(k)} + b^{(k)}$  for the $(k+1)$-th layer. Besides, the vector-valued function  $\phi^{(k)}: \mathbb{R}^{n_{k+1}} \to \mathbb{R}^{n_{k+1}}$ is applied element-wise to the pre-activation vector $z^{(k)}$, that is, 
$  \phi^{(k)}(z^{(k)}) =
    [
    \varphi(z^{(k)}_1), \cdots, \varphi(z^{(k)}_{n_{k+1}}) ]^T$, 
where $\varphi: \mathbb{R} \to \mathbb{R}$ is the (scalar) activation function.  Common activation choices are:  ReLU  $\varphi(z) = max(0,z)$; sigmoid $\varphi(z) = \frac{1}{1+e^{-z}}$; and tanh $\varphi(z) = tanh(z)$. 

The closed-loop system  with dynamics \eqref{eq:nonlinear_model} and a previously trained neural-network control policy \eqref{eq:NN_model}, is governed by
\begin{equation} \label{eq:closed_loop}
    x(t+1) = f_g(x(t),w(t)) =  f\big(x(t),g(x(t)),w(t) \big).
\end{equation}

Accordingly, given an initial set $\mathcal{X}_0 \subset \mathbb{R}^{n_x}$, the forward reachable set of \eqref{eq:closed_loop} at time step $t$  is denoted as $\mathcal{X}(t)$. For $t\geq 1$, this set is defined as: 
\begin{equation} \label{eq:reachable_set}
\begin{aligned}
          \mathcal{X}(t) = \big\{x(t) \, |& \,  \exists (x(0), w(0:t-1)) \in \mathcal{X}_0 \times \mathcal{W} \times ... \times \mathcal{W}, \\
       & \, \forall \tau \in [0, t-1], x(\tau +1) = f_g(x(\tau),w(\tau)) \big \}. 
\end{aligned}
\end{equation}


\subsection{Finite-time reach-avoid (RA) verification problem}

Given a goal set $\mathcal{G} \subset \mathbb{R}^{n_x}$ a sequence of avoid sets $\mathcal{A}(t) \subset \mathbb{R}^{n_x}$ and a finite time horizon $N \in \mathbb{N}^+$, it is desired to test whether  
\begin{equation}\label{eq:safety_specifications}
  \begin{aligned}
    &\mathcal{X}(N) \subseteq \mathcal{G} \\
    &\mathcal{X}(t) \cap \mathcal{A}(t) = \emptyset, \quad \forall t = 0,...,N-1
\end{aligned}  
\end{equation}
\noindent holds true for the closed loop system \eqref{eq:closed_loop}. In general, the exact evaluation of \eqref{eq:safety_specifications} for a NNCSS is a computationally intractable problem.  Thus, the problem is resorted to iteratively compute a tractable  over-approximation of the reachable set $\mathcal{X}(t) \subseteq \bar{\mathcal{X}}(t)$, to test \eqref{eq:safety_specifications} using $\bar{\mathcal{X}}(t)$ instead. Because of the over-approximation, the proposed verification setting only provides one-sided guarantees, that is, if  $\bar{\mathcal{X}}(t)$ satisfies \eqref{eq:safety_specifications} then it can be guaranteed that \eqref{eq:reachable_set} will satisfy the RA property, but no sound conclusion about the safety of \eqref{eq:reachable_set} can be made if the over-approximation $\bar{\mathcal{X}}(t)$ violates \eqref{eq:safety_specifications}. Therefore, the computation of tight over-approximations is of paramount importance, so that a maximum number of truly satisfied specifications can be computationally proven as such.


\section{Closed-loop verification using s-zonotopes} \label{sec:verification}

This section presents the methodology for computing a
sound over-approximation of the closed-loop system that preserves system-controller linear dependencies. The computation takes advantage of s-zonotopes described in the previous section. The abstraction of the control loop components using affine symbolic expressions is presented below.


\subsection{Initial set} \label{sec:initial_set}

It is assumed that the initial set can be described by the set-valued interpretation of an s-zonotope $\mathcal{X}_{|s}{(0)} = \langle c_0, R_0, I_0 \rangle_{|s}$, where $c_0\in \mathbb{R}^{n_x}$ and $R_0 \in \mathbb{R}^{n_x \times n_0}$ and  $I_{0} = !(n_0)$ is a set of $n_0$ unique identifiers  for the interval valued symbols $s_{I_0}$. In other words, it is assumed that $\mathcal{X}_0 = \mathcal{X}_{|\iota}(0)$. Note that, any arbitrary zonotopic set $\{c +R \xi \, | \,  \|\xi\|_\infty \leq 1\}$ can be abstracted as an s-zonotope by characterizing the independent behaviour of the generators through new interval type symbols.


\subsection{NN controller affine abstraction} \label{sec:controller_abs_linear}

For the sake of simplicity of notations, the temporal notation is dropped here. Given a state bounding s-zonotope $\mathcal{X}_{|s} = \langle c, R, I \rangle_{|s}$ and a NN controller \eqref{eq:NN_model}, the idea is to abstract the NN behavior through  an affine symbolic expression of the form \begin{equation} \label{eq:input_sym}
 \mathcal{U}_{|s}  = \langle C_u, [G, \,H], [I; \, J]\rangle_{|s} =c_u + Gs_I + Hs_{J},
\end{equation}
\noindent such that, it guarantees the local enclosure of the network outputs, i.e. $g(\mathcal{X}_{|\iota}) \subseteq \mathcal{U}_{|\iota}$. Note that,  expression \eqref{eq:input_sym} captures the linear dependencies of the state symbols (identified by $I$), plus the addition of new error symbols (identified by $J$) that are introduced to guarantee the soundness of the method. 

The computation of vector $c_u$, matrices $G,H$, and the identifiers vector $J$ is discussed below. The focus is on generating a dependencies-preserving inclusion for an arbitrary layer of NN \eqref{eq:NN_model}, since a sound enclosure for the whole network follows by induction due its sequential nature. For simplicity, the layer superscript is removed below and the superscript $^+$ is used to denote the next layer.

\textbf{Affine mapping} $ \ $ Given the  $s$-zonotope $\mathcal{X}_{|s} = \langle c, R, I \rangle_{|s}$, the affine mapping $\mathcal{Z}_{|s} = W\mathcal{X}_{|s} + b$ in the layers of \eqref{eq:NN_model} yields a (pre-activation) s-zonotope of the form
\begin{equation}
\begin{aligned}
    &\mathcal{Z}_{|s} = \langle \check{c}, \check{R},I \rangle_{|s},\\
    & \check{c} = Wc+b,    \quad \check{R} = WR.
\end{aligned}
\end{equation}

\textbf{Activation functions} $ \ $ Activation functions  $\varphi(\cdot)$ in \eqref{eq:NN_model}  are applied element-wise to the pre-activation vector. Hence, the projection of $\mathcal{Z}_{|s}$ onto  the $i$-th neuron, yields the s-zonotope
\begin{equation}\label{eq:projection}
    \mathcal{Z}_{i|s} = \langle \check{c}_i,\check{R}_{[i,:]},I \rangle_{|s}.
\end{equation}

Notice that, any point belonging to set-valued interpretation $\mathcal{Z}_{i|\iota}$ of \eqref{eq:projection} is confined within an interval $[l_i, \, u_i]$, where, since $\mathcal{Z}_{i|\iota}$ is a one dimensional zonotopic set, it follows that $\mathcal{Z}_{i|\iota} = [l_i, \, u_i]$ with the  lower and upper bounds
\begin{equation} \label{eq:bounds}
    l_i = \check{c}_i - \|\check{R}_{[i,:]}\|_1, \quad u_i = \check{c}_i + \|\check{R}_{[i,:]}\|_1.
\end{equation}

Therefore, the soundness of the method can be certified by guaranteeing the inclusion  (see \Cref{def:inclusion}) of the graph of the activation function in the range $[l_i, \, u_i]$. To that end, 
the activation function $\varphi(\cdot)$ is abstracted  through an affine symbolic function of the form
\begin{equation}\label{eq:neuron_post_activation}
  \mathcal{X}_{i|s}^{+} =   \alpha_i \mathcal{Z}_{i|s} + \beta_i   + \gamma_i s_{j},
\end{equation}
\noindent where  $s_j$ represents a new independent symbol (identified through $j = !(1)$) that must be introduced to guarantee the  full coverage of the activation function graph on the considered range, that is, in order to satisfy the  condition
\begin{equation} \label{eq:inclusion_condition}
  \begin{bmatrix}
       \mathcal{Z}_{i|\iota} \\
    \varphi( \mathcal{Z}_{i|\iota}) 
    \end{bmatrix} \subseteq  \begin{bmatrix}
       \mathcal{Z}_{i|\iota} \\
  \alpha_i \mathcal{Z}_{i|\iota}  + \beta_i   + \gamma_i \mathcal{D}(s_{j})
    \end{bmatrix}.
\end{equation}

The $i$-th neuron post-activation s-zonotope $\mathcal{X}_{i|s}^{+}$ in \eqref{eq:neuron_post_activation} not only guarantees that its set-valued interpretation encloses the neuron output, but it preserves the linear influence of the symbols $s_{I}$ in the output set. This later point plays a fundamental role since it allows to retain the interplay between the inputs of the neurons at the same layer. Coherently, the layer post-activation s-zonotope can be computed by vertically concatenating in a recursive fashion the different  $\mathcal{X}_{i|s}^{+}$ after  rewriting them using the same set of symbols
\begin{equation}
   \mathcal{X}_{|s}^+ = [ \, ... \, [\,[\mathcal{X}_{1|s}^{+};\mathcal{X}_{2|s}^{+}];\mathcal{X}_{3|s}^{+}] \, ... \, ;\mathcal{X}_{n_k|s}^{+}].
\end{equation}

\begin{myproposition}[NN s-zonotope] \label{prop:linear_abs}
Given the  s-zonotope $\mathcal{X}^{(0)}_{|s} = \langle c^{(0)}, R^{(0)}, I^{(0)}\rangle_{|s}$, and let $\alpha^{(k)}, \beta^{(k)}, \gamma^{(k)} \in \mathbb{R}^{n_k+1}$ be some parameter vectors that guarantee the inclusion of the $n_{k+1}$ activation functions in the $k$-th layer, then the enclosure of the NN output set $  g(\mathcal{X}^{(0)}_{|\iota}) \subseteq \mathcal{U}_{|\iota} = \langle c_u, [G,\, H], [I; \, J]\rangle_{\iota}$ is guaranteed for the s-zonotope in \eqref{eq:input_sym} with parameters
\begin{subequations}
\begin{align}
        &c^{(k+1)} = diag(\alpha^{(k)})(W^{(k)}c^{(k)} + b^{(k)}) + \beta^{(k)}, \label{eq:prop_center} \\
        & \tilde{H}^{(k+1)} = \begin{bmatrix}
        diag(\alpha^{(k)}) W^{(k)}\tilde{H}^{(k)}, & diag(\gamma^{(k)})
        \end{bmatrix}, \\
        &  \tilde{G}^{(k+1)} =   diag(\alpha^{(k)}) W^{(k)}\tilde{G}^{(k)}, \quad k = 1,...,l-1, \label{eq:prop_matrix} \\
        &c_u = W^{(l)}c^{(l)}+b^{(l)}, \label{eq:prop_center_out} \\
        &H = W^{(l)}\tilde{H}^{(l)},  \\
        &K = W^{(l)}\tilde{K}^{(l)}, \label{eq:prop_matrix_out} \\
        &J = [!(n_{1}); \ ...  ; !(n_{l})],   \label{eq:prop_ids} 
    \end{align}
\end{subequations}
where $\tilde{G}^{(1)} = diag(\beta^{(0)}) W^{(0)}R^{(0)}$ and $\tilde{H}^{(1)} =  diag(\gamma^{(0)})$.
\end{myproposition}

\begin{proof}
Expressions \eqref{eq:prop_center}-\eqref{eq:prop_matrix_out} result from the recursive application of  \Cref{lem:common_id} and the vertical concatenation of the  post-activation s-zonotopes \eqref{eq:neuron_post_activation} for the $n_{k+1}$ neurons of the $k$-th layer. Besides, \eqref{eq:prop_ids} reflects the symbols identifier update of the noise terms introduced at the neurons of each layer.

Regarding the output inclusion, starting with an initial set $\mathcal{X}^{(0)}_{|\iota}$, by induction, given the pre-activation s-zonotope $\mathcal{X}^{(k)}_{|s}$, the operations at the $k$-th layer are: affine mapping; linear abstraction (inclusion preserving for appropriate triplet ($\alpha_i^{(k)},\beta_i^{(k)},\gamma_i^{(k)}$)); and vertical concatenation. Thus, the composition of inclusion functions being an inclusion function, the proof follows.
\end{proof}

For each neuron, the triplet of parameters $(\alpha, \beta, \gamma)$ must be  appropriately designed to satisfy \eqref{eq:inclusion_condition}, while minimizing the conservatism induced by using an affine relaxation. In this regard, a relevant heuristic consists in minimizing the magnitude of the error symbol introduced to guarantee the activation function graph enclosure, i.e. to minimize $|\gamma|$. Due to the independent behaviour of the error symbol, this can be reformulated as minimizing the area of the enclosing parallelogram~\cite{singh2018fast}.

\begin{mylemma}\label{lem:triplets}
Given the bounds $[l, \, u]$ in \eqref{eq:bounds} with $l < u$, the triplet of parameters $(\alpha^*, \beta^*, \gamma^*)$ that minimizes $|\gamma|$ while guaranteeing the satisfaction of \eqref{eq:inclusion_condition} are:
\begin{itemize}[leftmargin =*]
    \item ReLU function  $\varphi(x) = max(0,x)$ 
    \begin{equation}
     \alpha^* = \frac{\varphi(u)-\varphi(l)}{u-l}, \quad \beta^* = \gamma^* = \frac{\varphi(l)-\alpha^* \cdot l}{2}.
\end{equation}
    \item S-shaped functions
    \begin{itemize}
        \item   Sigmoid $\varphi(x) = \frac{1}{1+e^{-x}}$ with $\varphi^\prime(x) = \varphi(x)(1-\varphi(x))$
        \item tanh $\varphi(x) = tanh(x)$ with $\varphi^\prime(x) = 1-\varphi(x)^2$
    \end{itemize}
    \vspace*{2mm}
    \begin{equation}
        \begin{aligned}
         \alpha^* &= \min(\varphi^\prime(l),\varphi^\prime(u)), \\
         \beta^* &= \frac{\varphi(u)+ \varphi(l) - \alpha^* \cdot  (u+l)}{2}, \\
        \gamma^* &= \frac{\varphi(u) - \varphi(l) - \alpha^* \cdot (u -l)}{2}.
    \end{aligned}
    \end{equation}
\end{itemize}
\end{mylemma}

\begin{myremark}
The proposed NN abstraction method shares a similar structure with the zonotope abstraction based on affine arithmetic presented in \cite{singh2018fast} for the (open-loop) NN output bounding. However, here, the explicitly computed  affine symbolic expression \eqref{eq:input_sym} will further play a key role in closed-loop verification, and an efficient computation of the projection matrices exploiting the generators (matrix) structure of s-zonotopes, is also used.
\end{myremark}


\subsection{Dynamical system affine abstraction} \label{sec:system_abs_linear}

Similar to the NN controller dynamics \eqref{eq:NN_model}, the  function \eqref{eq:nonlinear_model} that describes the state evolution at time $(t+1)$ can be abstracted by means of an (inclusion preserving) affine mapping. The resulting s-zonotope will depend on  the symbols that define the state at time $t$, plus some extra symbols that account for:  I)  NN controller non-linearities; II) abstract system non-linearities; III) the uncertainty sources. 

For the computation of a state bounding s-zonotope, it is assumed that the function $f(\cdot)$ in \eqref{eq:nonlinear_model} results from the composition of elementary functions and operators for which an affine symbolic expression (s-zonotope) can be computed. Note that this is not much restrictive since (linear) operations such as linear image, sum or vertical concatenation are closed (i.e. they return s-zonotopes) under affine mappings. Besides, any univariate locally continuous differentiable function can be abstracted through an affine mapping.

\begin{mylemma} \label{lem:general_inclusion}
Let $h:[l, \, u] \to \mathbb{R}$ be a class $\mathcal{C}^1$  function on a given interval $[l, \, u] \subset \mathbb{R}$. Then, the function $\tilde{h}(x,\epsilon) = \alpha x + \beta + \gamma \epsilon$ satisfies that $\forall x \in [l, \, u], \ \exists \epsilon \in [-1, \, +1], h(x) = \tilde{h}(x,\epsilon)$ for the triplet of parameters:
\begin{equation*}
\begin{aligned}
        \alpha &= \frac{h(u)-h(l)}{u-l}, \quad \beta = \frac{h(\underline{x})+h(\bar{x})-\alpha(\underline{x}+\bar{x})}{2}, \\ \gamma &= \frac{h(\bar{x})-h(\underline{x})+\alpha(\underline{x}-\bar{x}) }{2},
\end{aligned}
\end{equation*}
where, defining $\xi(x) = h(x)-\alpha x$, then
\begin{equation*}
\begin{aligned}
    \bar{x} = \argmax_{x\in \{\delta_1,...,\delta_n,u\}}\xi(x),
     \quad \
    \underline{x} = \argmin_{x\in \{\delta_1,...,\delta_n,u\}}\xi(x),
\end{aligned}
\end{equation*}
with $\delta_1,...,\delta_n$  the stationary-points of $\xi(\cdot)$ in $[l, \, u]$. 
\end{mylemma}

\begin{proof}
See Appendix \ref{app:appendix_a}.
\end{proof}

\Cref{lem:general_inclusion} provides a method to propagate (inclusion preserving) s-zonotopes through univariate non-linearities. Besides, the approach in  \Cref{lem:general_inclusion} returns an optimal, in the sense of minimizing the magnitude $|\gamma|$ of the error symbol, set of parameters for convex/concave differentiable functions~\cite{combastel2020functional}. On the other hand, the interaction between multiple variables can be handled through the sum operation \eqref{eq:sum} or by over-approximating the product of two s-zonotopes.

\begin{mylemma}\label{lem:multiplication}
Given two 1-D s-zonotopes $\mathcal{X}_{|s} = \langle c_x, r^T, K\rangle_{|s}$ and $\mathcal{Y}_{|s}= \langle c_y, g^T, K\rangle_{|s}$ with a common set of symbols $s_K$ (with $n = card(s_K)$), then the product $\mathcal{X}_{|s} \times \mathcal{Y}_{|s}$ is included by the s-zonotope $\mathcal{L}_{|s} = \langle c_l, [l^T, \, m], [K; \, j]\rangle_{|s}$  with
\begin{equation}
\begin{aligned}
   c_l &= c_xc_y + \frac{1}{2}\sum_{i=1}^nr_ig_i, \quad \quad  l = c_xg + c_yr, \\
   m &= \frac{1}{2}\sum_{i=1}^n|r_ig_i| + \sum_{i=1}^n\sum_{l>i}^n |r_ig_l + r_lg_i|,
\end{aligned}
\end{equation}
\noindent and $j = !(1)$.
\end{mylemma}

\begin{myexample}
Consider the system  $x^+ = \sin(x) - u + 0.1w$, with an initial set $\mathcal{X}_0 = [0, \, 1]$ described by $\mathcal{X}_{|s,0} =  0.5 + 0.5s_1$. This system is controlled by a NN with 1 layer of 2 neurons that, for $\mathcal{X}_{0}$, is abstracted as 
 $\mathcal{U}_{|s}  = 0.1 + 0.2s_1 -0.1s_2 + 0s_3$. The non-linear function $h(x) = \sin(x)$ is abstracted, for $\mathcal{X}_{0}$, as
$\hat{\mathcal{X}}_{|s}  = 0.45 + 0.42 s_1 + 0.03 s_4$. Besides, the independent behaviour of the disturbances is captured by $\mathcal{W}_{|s} = s_5$.  Accordingly, a dependency preserving over-approximation of the successor state  is given by $\mathcal{X}^+_{|s} = 0.35 + 0.22s_1+0.1s_2+0.03s_4+0.1s_5$. The a priori knowledge on the number of error symbols introduced at each abstraction (e.g.  $\mathcal{U}_{|s}$ introduces up to two error symbols, one per neuron) allows to directly store the generators matrix of each s-zonotope by taking into account the common set of symbols, thus  providing an efficient computation of required operations as shown below:
\begin{equation*}
\renewcommand\arraystretch{1.3}
\begin{array}{ccl}
& \begin{array}{>{\centering\arraybackslash$} p{0.6cm} <{$}  >{\centering\arraybackslash$} p{0.5cm} <{$} >{\centering\arraybackslash$} p{0.5cm} <{$}  >{\centering\arraybackslash$} p{0.5cm} <{$} >{\centering\arraybackslash$} p{0.5cm} <{$}  >{\centering\arraybackslash$} p{0.5cm} <{$} }
    c & s_1 & s_2 & s_3 & s_4 & s_5 
\end{array}  & \\[5pt] 
- &\mleft[ \begin{array}{>{\centering\arraybackslash$} p{0.6cm} <{$} | >{\centering\arraybackslash$} p{0.5cm} <{$} >{\centering\arraybackslash$} p{0.5cm} <{$}  >{\centering\arraybackslash$} p{0.5cm} <{$} >{\centering\arraybackslash$} p{0.5cm} <{$}  >{\centering\arraybackslash$} p{0.5cm} <{$} }
  0.1 & 0.2 & -0.1 & 0 & 0 & 0 
\end{array} \mright] & \mapsfrom \ \mathcal{U}_{|s}  \\[5pt]
+& \mleft[ \begin{array}{>{\centering\arraybackslash$} p{0.6cm} <{$} | >{\centering\arraybackslash$} p{0.5cm} <{$} >{\centering\arraybackslash$} p{0.5cm} <{$}  >{\centering\arraybackslash$} p{0.5cm} <{$} >{\centering\arraybackslash$} p{0.5cm} <{$}  >{\centering\arraybackslash$} p{0.5cm} <{$} }
    0.45 & 0.42 & 0 & 0 & 0.03 & 0 
\end{array} \mright] & \mapsfrom \ \hat{\mathcal{X}}_{|s}\\[5pt] 
+ &\mleft[ \begin{array}{>{\centering\arraybackslash$} p{0.6cm} <{$} | >{\centering\arraybackslash$} p{0.5cm} <{$} >{\centering\arraybackslash$} p{0.5cm} <{$}  >{\centering\arraybackslash$} p{0.5cm} <{$} >{\centering\arraybackslash$} p{0.5cm} <{$}  >{\centering\arraybackslash$} p{0.5cm} <{$} }
  0 & 0 & 0 & 0 & 0 & 0.1
\end{array} \mright] & \mapsfrom \ 0.1 \cdot \mathcal{W}_{|s} \\[5pt]
= & \mleft[ \begin{array}{>{\centering\arraybackslash$} p{0.6cm} <{$} | >{\centering\arraybackslash$} p{0.5cm} <{$} >{\centering\arraybackslash$} p{0.5cm} <{$}  >{\centering\arraybackslash$} p{0.5cm} <{$} >{\centering\arraybackslash$} p{0.5cm} <{$}  >{\centering\arraybackslash$} p{0.5cm} <{$} }
  0.35 & 0.22 & 0.1 & 0 & 0.03 & 0.1
\end{array} \mright] & \mapsfrom \ \mathcal{X}^+_{|s}
\end{array}   
\end{equation*}
\end{myexample}

\begin{algorithm}[b!] 
\caption{Finite-time RA verification}\label{alg:reachability}
\hspace*{1mm} \textbf{Inputs:} $\mathcal{X}_0$, NN param $(\bm{W},\bm{b})$,  $\mathcal{G}, \mathcal{A}(i)$, $N$,  $q$ \\
\hspace*{1mm} \textbf{Outputs:} isRAok, $t_{err}$,  $\mathcal{X}_{|s}(j)$ $j = 0,...,min(t_{err},N)$.
\begin{algorithmic}[1]
\State{\textbf{Initialize:} Generate $\mathcal{X}_{|s}(0)$; set $t_{err} \gets \infty$}
\For{i = $0$ to $N-1$}
\If{$\mathcal{X}_{|\iota}(i) \cap \mathcal{A}(i)$}
\State {$t_{err} \gets i$}
\State{\textbf{break all}}
\Else
\State{$\mathcal{U}_{|s}(i) \gets controller(\mathcal{X}_{|s}(i),\bm{W},\bm{b})$} 
\State{$\bar{\mathcal{X}}_{|s}(i+1) \gets system(f(\cdot),\mathcal{X}_{|s}(i),\mathcal{U}_{|s}(i), \mathcal{W}_{|s}(i))$}
\State{$\mathcal{X}_{|s}(i+1) \gets \downarrow_q \bar{\mathcal{X}}_{|s}(i+1)$}
\If{$(i == N-1) \land (\mathcal{X}_{|s}(N) \not\subseteq \mathcal{G})$}
\State {$t_{err} \gets i+1$}
\EndIf
\EndIf
\EndFor
\State isRAok = ($t_{err} == \infty$)
\end{algorithmic}
\end{algorithm}


\subsection{Closed-loop integration}

\Cref{alg:reachability} describes the main steps for the closed-loop finite-time reach-avoid verification problem under an s-zonotope formulation. Steps 7 and 8 represent the local abstraction of the NN controller and dynamical system  through an affine symbolic expression as described in \Cref{sec:controller_abs_linear} and \Cref{sec:system_abs_linear}, respectively. Note that, in Step 8, the uncertain behaviour of disturbances is adequately modeled by  generating a set of  independent symbols at each call $\mathcal{W}_{|s}(i) = s_{I_{w}}$, where $I_w = !(n_w)$. Besides, Step 9 includes the reduction operator introduced in \Cref{def:reduction}. At this step, the less relevant symbols, that is, those that have the least significant impact at the current time instant s-zonotope,  are truncated consistently with the tuning of $q$ providing control on the trade-off between computational complexity and accuracy while preserving inclusion. 

\begin{myremark} \label{rem:discrete_time}
Symbolic approaches also allow to efficiently handle shorter discretization periods ($\Delta T$) in the (discrete-time) system model \eqref{eq:nonlinear_model} than the controller update period ($\Delta h$), since the dependencies between control inputs repeated at different time steps are preserved. Therefore, the discretization period ($\Delta T$), and thus the discretization error,  can be made smaller (up to reduction) at the cost of increasing the number of iterations $N$ for which the system should be evaluated  to meet a specified time horizon of $N \cdot \Delta T$. As an example: consider $\Delta T = \Delta h/2$ and the first two iterations of a system with held input over $\Delta h$ given by $x(1) = -x(0) + u(0)$ and $x(2) = -x(1) + u(0)$, with  $x(0) \in  \langle 0,1,1\rangle_{\iota}$ and $u(0) \in \langle 0,1,2\rangle _{\iota}$. A classical set-valued evaluation yields $\mathcal{X}_{|\iota}(1)  = - \langle 0,1,1\rangle _{\iota} +  \langle 0,1,2\rangle _{\iota} = [-2,2]$ and $\mathcal{X}_{|\iota}(2) = -\mathcal{X}_{|\iota}(1) + \langle 0,1,2\rangle _{\iota} = [-3,3]$, whereas operating at symbolic level gives $\mathcal{X}_{|s}(1) = \langle 0,[-1, \, 1],[1, \, 2]\rangle_{|s}$ and $\mathcal{X}_{|s}(2) = -\mathcal{X}_{|s}(1) + \langle 0,1,2\rangle_{|s} = \langle 0,1,1\rangle_{|s}$ yielding $\mathcal{X}_{|\iota}(2) =[-1, 1]$.
\end{myremark}


\section{Polynomial symbolic expressions} \label{sec:polynotope}

The methodology presented in \Cref{sec:verification} based on s-zonotopes can  be readily extended to use any other well-formed symbolic expression as long as its set-valued interpretation guarantees the inclusion of the controller-system output sets. In particular, this section investigates the use of symbolic polynotopes (s-polynotopes) as in \cite{combastel2020functional} to compute sound approximations relying on polynomial rather than affine dependencies. To that end, s-polynotopes are briefly recalled in \Cref{sec:s-polynotopes}. Then, \Cref{sec:polynotope_NN} investigates the abstraction of the I/O mapping of a NN through an inclusion preserving polynomial symbolic function computed in a propagation-based fashion, that is, in a compositional way possibly benefiting from a systematic use of basic operator overloading for the sake of simple and generic implementations. Besides, \Cref{sec:polynotope_dynam} discusses the main aspects to address a finite-time RA verification problem using a s-polynotope formulation.


\subsection{Symbolic polynotopes} \label{sec:s-polynotopes}

Symbolic polynotopes enable a  tractable computational representation of polynomial symbolic functions by encoding all their relevant information (generators, symbol identifiers and order of the monomials) into matrices.

\begin{mydefinition}[s-polynotope \cite{combastel2020functional}]\label{def:polynotope}
A symbolic polynotope $\mathcal{P}_{|s}$ is a polynomial function that can be written in the form $c+Rs_I^E$ where vector $c$ and matrices $R$ and $E$ do not depend on the symbolic variables in $s_I$. Notation: $\mathcal{P}_{|s} = \langle c,R,I,E\rangle_{|s} = c+Rs_I^E$.
\end{mydefinition}

\Cref{def:polynotope} uses the exponential matrix notation (as in Definitions 23-25 in \cite{combastel2020functional}), where the usually sparse matrix $E$ accounts for exponents of the symbols involved in each monomial. As an example, $s_I = [s_1, s_2]^T$ and $E = [1 \, 0 \, 3; 0 \, 2 \, 4]$, yields $s_I^E = [s_1, s_2^2, s_1^3s_2^4]^T$. Similar to s-zonotopes, a distinction is made between an s-polynotope  as defined in \Cref{def:polynotope} and its set-valued interpretation  defined as the (possibly non-convex) set $\mathcal{P}_{|\iota} =  \{c+R\sigma^E \, | \, \sigma \in \mathcal{D}(s_I)\}$.

Symbolic polynotopes obviously extend  s-zonotopes (obtained from $E = \mathcal{I}$, i.e. with identity as exponent matrix) and are closed under the extension of basic operations already defined for s-zonotopes like linear image, sum or concatenation. The reader is referred to \cite{combastel2020functional} for further details on how to define and operate on s-polynotopes.


 \subsection{NN controller polynomial abstraction} \label{sec:polynotope_NN}
 
The abstraction of the I/O map of a NN controller of the form \eqref{eq:NN_model} using s-polynotopes is presented below. In particular, given a state bounding s-polynotope $\mathcal{X}_{|s} = \langle c,R,I,E\rangle_{|s}$ (note that any initial s-zonotope in \Cref{sec:initial_set} can be directly transformed into an equivalent s-polynotope) the idea is to compute a polynomial symbolic map of the form
\begin{equation}\label{eq:polynomal_map}
    \mathcal{U}_{|s} = \langle c_u, G, Q, E_u\rangle =  c_u + Gs_Q^{E_u},
\end{equation}
\noindent such that, the enclosure of the network outputs is guaranteed, i.e. $g(\mathcal{X}_{|\iota}) \subseteq \mathcal{U}_{|\iota}$. The vector of identifiers in \eqref{eq:polynomal_map} has the structure  $Q = [I, \, J]$, thus involving the symbols in the state bounding s-polynotope (identified by $I$) as well as error symbols (identified by $J$). Notice that the exponent matrix $E_u$ may also capture cross terms involving symbols with identifiers in both $I$ and $J$.

Similar to \cref{sec:controller_abs_linear}, the computation of \eqref{eq:polynomal_map} can be obtained from a forward propagation of $\mathcal{X}_{|s}$ through (in this case) a polynomial relaxation of the activation functions. The pre-activation s-polynotope $\mathcal{Z}_{|s} = W\mathcal{X}_{|s} +b$ and its projection onto the $i$-th neuron are given by
\begin{equation}
\begin{aligned}
           \mathcal{Z}_{|s} &= \langle \Breve{c},\Breve{R}, I,E \rangle_{|s}, \quad  &\mathcal{Z}_{i|s} &= \langle \Breve{c}_i,\Breve{R}_{[i,:]}, I,E \rangle_{|s}, \\
           \Breve{c} &= Wc+b, \quad &\Breve{R} &= WR.
\end{aligned}
\end{equation}

Bounding the set-valued interpretation of  $\mathcal{Z}_{i|s}$ within the interval $[l_i, \, u_i]$, then the polynomial structure of s-polynotopes enables to obtain a sound over-approximation of the NN output by locally covering the activation function graph through an $n$-order polynomial expression of the form
\begin{equation} \label{eq:n-pol}
    \mathcal{X}_{i|s}^{+} = \sum_{m=1}^n \alpha_{i,m}(\mathcal{Z}_{i|s})^m + \beta_i + \gamma_i s_j,
\end{equation}
\noindent where  $s_j$ represents a new independent symbol (identified through $j = !(1)$) introduced to guarantee the  enclosure of the activation function graph in the  range $[l_i, \, u_i] $. Therefore, since $ \mathcal{X}_{i|s}^{+}$ results in an s-polynotope arising from the polynomial mapping of s-polynotopes, the layer post-activation s-polynotope $ \mathcal{X}_{|s}^{+}$ is computed by vertically concatenating the neuron post-activation s-zonotopes.

Polynomial over linear abstractions of the activation function not only allow to reduce the conservatism introduced by the error symbols, but also enable the to compute input-output symbolic relationships that better fit the activation behaviour.

\begin{myexample}\label{ex:s-polynotope}
Suppose that the projection onto a ReLU neuron  is given by the s-polynotope $ \mathcal{Z}_{|s} = 0.5 -0.5s_1 + s_1s_2$, whose set-valued interpretation is bounded/included in the interval $ [l, \, u] = [-1, \, 2]$. Then, the ReLU function can be locally abstracted over this range using an $n=2$-order polynomial of the form $\mathcal{X}_{|s}^+ = \alpha_{2}(\mathcal{Z}_{|s})^2 + \alpha_{1} \mathcal{Z}_{|s} + \beta + \gamma s_3$, with  $(\alpha_{2}, \ \alpha_{1}) = (0.25, \, 0.5), \ \beta = \gamma = 0.125$. This, in turn, generates the post-activation s-polynotope 
\begin{equation*}
\begin{aligned}
        \mathcal{X}_{|s}^+ &= 0.25(0.5-0.5s_1+s_1s_2)^2 + 0.5(0.5-0.5s_1+s_1s_2) \\ & \ \  + 0.125 + 0.125s_3 \\
         & = 0.4375 -0.375s_1 +0.0625s_1^2 + 0.125s_3 + 0.75s_1s_2 \\ & \ \ -0.25s_1^2s_2 + 0.25s_1^2s_2^2.
\end{aligned}
\end{equation*}
\noindent \Cref{fig:poly_relu} depicts the non-convex local  enclosure generated by the set-valued interpretation of $[\mathcal{Z}_{|s};\mathcal{X}_{|s}^+]$.
\end{myexample}  

\begin{figure}[t!]
    \centering
    \includegraphics[width=\linewidth]{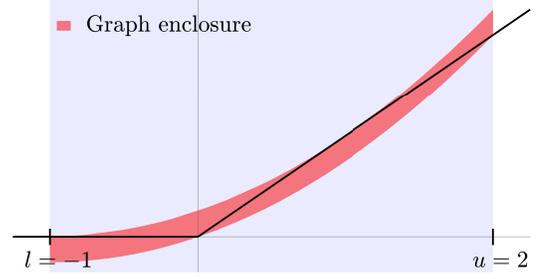}
    \caption{2-nd order polynomial enclosure of ReLu function (\Cref{ex:s-polynotope})}
    \label{fig:poly_relu}
\end{figure}

\Cref{ex:s-polynotope} evidences the complexity/accuracy trade-off inherent to using $n$-order polynomial abstractions: a high $n$ grants an accurate representation of the activation functions; whereas, on the other hand, it increases the computational complexity  due to the increased number of monomials. A reduction strategy is thus used to manage the representation complexity. It can consist in either truncating the maximum degree of the polynomial approximation \eqref{eq:polynomal_map},  or limiting the maximum number of monomials involved. To address this latter issue, an approach consists in (independently) assessing the monomials relevance  based on the $2$-norm of the generators (matrix columns)~\cite{combastel2005state}, and use natural interval extension~\cite{jaulin2001interval} to bound the list of selected monomials through a reduced number of independent symbols. 

The ability of the final s-polynotope \eqref{eq:polynomal_map} to generate a sound over-approximation of the network outputs is guaranteed by selecting (for each neuron) a triplet  $(\alpha, \beta, \gamma)$, where  $\alpha = [\alpha_n,...,\alpha_1]^T$ is an $n$-dimensional vector, that ensures the (local) coverage of the activation function. In the case of the commonly used ReLU activation functions, as shown in \Cref{ex:s-polynotope} their convex nature allows describe them more accurately  than with the sole affine dependencies by using $2$-nd order polynomial expressions. The reduction in the magnitude related to the error symbol is especially significant in those situations where an affine approximation yields a rough description of ReLU function, i.e. for $|l| \approx u$ (if $l<0<u$). 

\begin{myproposition}\label{prop:triplet_param_quad}
Given the interval $[l, \, u]$ with $l <0< u$ and a ReLU activation function $\varphi(x)= max(x,0)$. The  set of parameters
{\begin{itemize}[leftmargin=3mm]
        \item[]  $\alpha_2  =   \frac{1}{2u}, \ \alpha_1 = 1-\alpha_2u, \ \beta =  \gamma= \frac{\alpha_2u^2}{8}$ \hspace{1mm} if $|l| \leq  u \leq 2|l|$       
        \item[]  $\alpha_2  =   \frac{-1}{2l}, \ \alpha_1 = -\alpha_2l, \   \beta =  \gamma= \frac{\alpha_2l^2}{8}$  \hspace*{5.5mm} if $u < |l| \leq 2u$ 
\end{itemize}\small}
\noindent guarantees that  $\eta(x,\epsilon) = \alpha_2x^2+\alpha_1x+\beta +\gamma\epsilon$ satisfies that $\forall x \in [l, \, u], \exists \epsilon [-1, \, 1], \ \varphi(x) = \eta(x,\epsilon)$ with $|\gamma| \leq \frac{3}{8} |\gamma_{aff}^*|$, where $\gamma_{aff}^*$ is the error symbol introduced by the  affine abstraction in \Cref{lem:triplets}.
\end{myproposition}
\begin{proof}
See Appendix \ref{app:appendix_b}.
\end{proof}


\subsection{Finite-time RA using s-polynotopes} \label{sec:polynotope_dynam}
 
The main aspects in addressing the RA verification problem using s-polynotopes are discussed below. In general, the same steps presented in \Cref{alg:reachability} can be used while adapting the NN controller and dynamical system abstraction to an s-polynotope formulation. In this case, the computation of the s-polynotope $\mathcal{U}_{|s}$ in Step 7 of \Cref{alg:reachability} has already been presented in \cref{sec:polynotope_NN}. Regarding the abstraction of the non-linear function $f(\cdot)$ in Step 8, since s-polynotopes constitute an extension of s-zonotopes, $f(\cdot)$ can always be abstracted using (at least)  the affine dependency preserving method in  \Cref{lem:general_inclusion}. Note that, s-polynotopes also enable the description of (multivariate) polynomial equations without the need to over-approximate them, at least in all intermediate symbolic compositions and up to some tunable computation load.

It must be taken into account that several operations in a s-polynotope formulation of \Cref{alg:reachability} such as bounding the projection of an s-polynotope onto a neuron, intersection/inclusion of an s-polynotope with an avoid/reach set or the reduction operator (in Step 9), in turn require the  computation of interval bounds from (multivariate) interval polynomial expressions. If computationally affordable, the  range bounds computed by a (simple and fast) interval extension may be refined either by iteratively bisecting the variables domain, or resorting to numerically reliable optimization-based methods.


\section{Input partitioning strategy} \label{sec:partitioning}

In general, the conservatism induced by abstraction-based verification tools strongly depends on the size of the initial set. It is thus extremely useful to assess the regions of the initial/input space for which meaningful (i.e. not too coarse) over-approximations of the closed-loop system evolution can be obtained. To that end, this section  presents an algorithm to split the initial set of a NNCSs verification problem in a smartly guided way. More precisely, the proposed splitting strategy relies on and benefits from the dependency modeling and tracing used in \Cref{sec:verification}. In particular, the algorithm assesses the sensitiveness of the initial/input directions on the satisfaction of a safety property by an s-zonotopic over-approximation through the analysis of the relative influence of the initial symbols. The principle of the algorithm is introduced in \Cref{sec:general_idea_part}. Then, some relevant notions are detailed in \cref{sec:related_notions}, whereas the algorithm pseudo-code for a RA problem implementation is reported in \Cref{sec:implementation_part}. Finally, some further discussion on different settings is presented.

\vspace*{-2mm}
\subsection{Splitting principle}\label{sec:general_idea_part}

The main idea of the proposed algorithm is to keep a linear increase in the number of subsets by only splitting at each iteration the sole initial/input set symbol that has greater influence on the satisfaction of the safety property $\mathcal{S}$ to be verified. To that end, notice that given any initial s-zonotope $\mathcal{X}_{|s}(0) = \langle c_0, R_0, I\rangle_{|s}$ as defined in \Cref{sec:initial_set}, then the successive computation of forward reachable sets returns
 over-approximating s-zonotopes structured as
\begin{equation}\label{eq:out_alg1}
   \mathcal{X}_{|s}(t) =   c_f + R_fs_I + G_fs_J, 
\end{equation}
\noindent where the matrix $R_f$ (resp. $G_f$) reflects the impact of the initial (resp. error) symbols identified by $I$ (resp. $J$) on the computed over-approximation at time $t$. Typically, testing $\mathcal{S}(\mathcal{X}_{|s}(t) )$ boils down to a metric/size evaluation on the set-valued interpretation of $\mathcal{X}_{|s}(t) $ (e.g. to check threshold trespassing). Hence, due to the linearity of \eqref{eq:out_alg1}, the influence of each input symbol $s_i$  $(i\in I)$ can be assessed using a metric that gauges the generator (column of $R_f$) size that is related to $s_i$, whereas the accuracy of an s-zonotope approximation to evaluate $\mathcal{S}$ can be determined by measuring the zonotope $\langle 0, G_f\rangle$ spanned by the error symbols. 

Therefore, at each iteration of the algorithm, an input s-zonotope $\mathcal{X}_{|s}(0)$, such that the corresponding output/final s-zonotope does not satisfy $\mathcal{S}$, is split into two new input s-zonotopes that are later evaluated on the satisfaction of $\mathcal{S}$. The algorithm may run until the satisfaction of the safety property, or until the accuracy of the method (gauged through $\langle 0, G_f\rangle$) is below a certain threshold.


\subsection{Relevant notions} \label{sec:related_notions}

Some relevant notions for the s-zonotope based partitioning algorithm are discussed below.

\subsubsection{Accuracy assessment} considering the evaluation of a safety property for a s-zonotope of the form \eqref{eq:out_alg1}, the accuracy of the over approximation can be assessed by gauging the  zonotope  $\langle 0, G_f\rangle$ spanned by a (set-valued) interpretation of the error symbols. In particular, further implementations make use of the zonotope $F$-radius \cite{combastel2015zonotopes}, that is,the Frobenius norm of the generators matrix $\|G_f\|_F$ to reasoning upon the quality of the affine approximation.

\subsubsection{Input symbols relative influence} the sensitivity of an input symbol $s_i\  (i\in I)$ is computed based on the F-radius ratios of the I/O zonotopes spanned by  $\iota s_i$. That is,  through the ratio ${\|R^{[i]}_f\|_2}/{\|R^{[i]}_0\|_2}$ where $R^{[i]}_0$ (and $R^{[i]}_f$) denote the columns of $R_0$ (and $R_f$) that multiply the symbol $s_i$. This relation is used to quantify how a variation on $s_i$ at the input s-zonotope $\mathcal{X}_{|s}(0)$ affects the output s-zonotope.

\subsubsection{Symbol bisection} bisecting a unit interval symbol $s_i\  (i\in I)$ is done by rewriting it as $s_i \to 0.5 + 0.5s_j$ and $s_1 \to -0.5 + 0.5s_k$, where $j = !(1)$ and $k = !(1)$, thus generating two new s-zonotopes.

\subsubsection{Polyhedral RA sets} checking the empty intersection and/or the inclusion of a state bounding set with/within a polyhedron in half-space representation of the type $\{h_i^Tx\leq r_i, \ i= 1,...,m\}$ can be done by evaluating the infimum/supremum of the projections of the bounding set onto the directions  $h_i \in \mathbb{R}^{n_x}$~\cite{kolmanovsky1998theory}. For  a state bounding s-zonotope of the form \eqref{eq:out_alg1}, the supremum of the dot product with $h$ is computed as
\begin{equation}\label{eq:sup_funct}
    \sup_{x \in \mathcal{X}_{|\iota}(t)} h^T x=  h^Tc_f + \|h^TR_f\|_1 +  \|h^TG_f\|_1 .
\end{equation}


\subsection{Algorithm implementation for finite-time RA} \label{sec:implementation_part}

\Cref{alg:input_partitioning} reflects the pseudo-code of the proposed partitioning strategy to check the satisfaction of a RA problem over a time horizon $N$. \Cref{alg:input_partitioning} uses square brackets to label the different s-zonotopes that arise after splitting. As an example, $\mathcal{X}_{|s}[0]$ reads as the initial s-zonotope (i.e. $\mathcal{X}_{|s}[0] = \mathcal{X}_{|s}^0$), which is split onto a second $\mathcal{X}_{|s}[1]$ and a third $\mathcal{X}_{|s}[2]$ s-zonotopes (with $\mathcal{X}_{|\iota}[1] \cup \mathcal{X}_{|\iota}[2] = \mathcal{X}_{|\iota}[0]$). Besides, $L$ denotes a set of integer labels/indices (for the above example $L = \{0,1,2\}$), and $\mathcal{X}_{|\iota}[L]$ is a shorthand for the set of s-zonotopes $\{\mathcal{X}_{|s}[l] \, | \,  l \in L\}$.

At each iteration, the routine \texttt{reach} runs a slightly modified version of \Cref{alg:reachability}, that, in this case, returns the last time instant (and the corresponding s-zonotope) for which the RA problem is not satisfied. These times-to-last-error are managed by vector $T$. The algorithm iteratively selects the label of the initial s-zonotope that yields the largest time-to-last-error (Step 14). The use of this backward management of the information that prioritizes to split
until 
the RA constraints are satisfied at time  $t$, then at time $t -1$, etc., will be further discussed in the next paragraph. Once the $l$-th (with $l \in L$) s-zonotope has been selected, \texttt{sym-select} returns the initial symbol identifier $i \in I$ that has greater relative influence over the violated property. The symbol $s_i$ of the $l$-th set is split by the routine \texttt{sym-split} that returns two new initial subsets (Step 6).  The times-to-last-error for the new s-zonotopes are computed and the set $L$ and vector $T$ updated (Steps 7-11). In particular, \Cref{alg:input_partitioning} runs either until the RA problem is satisfied for the whole set $L$, or until a maximum number $n_{max}$ of splits is reached.

\begin{algorithm}[t!] 
\caption{Input partitioning for finite-time RA }\label{alg:input_partitioning}
\hspace*{1mm}\textbf{Input:} same as \Cref{alg:reachability}, $n_{max}$  \\
\hspace*{1mm}\textbf{Output:} isRAok, set of s-zonotopes $\mathcal{X}_{|s}[L]$
\begin{algorithmic}[1]

\State{\textbf{Initialize:} $l = n = 0$; $L = \{l\}$; $\mathcal{X}_{|s}[0] = \mathcal{X}_{0}$}
\State{$(t,\mathcal{X}_{|s}(t)[0]) \gets \texttt{reach}(\mathcal{X}_{|s}[0],N)$}
\State $T \gets append(t)$
\While{$(\max(T)>0) \ \lor \ (n/2 == n_{max}$)}
\State{$i \gets \texttt{sym-select}(\mathcal{X}_{|s}[l], \mathcal{X}_{|s}(T(l))[l]) $}
\State{($\mathcal{X}_{|s}[n+1],\mathcal{X}_{|s}[n+2]) \gets \texttt{sym-split}(\mathcal{X}_{|s}[l], i) $}
\For{$j = 1$ to $2$}
\State{$(t,\mathcal{X}_{|s}(t)[n+j]) \gets \texttt{reach}(\mathcal{X}_{|s}[n+j],\max(T))$} %
\State $L \gets L \cup \{n+j\}$ 
\State $T \gets append(t)$
\EndFor
\State{$L \gets L \setminus \{l\}$} 
\State {$T(l) \gets 0$}
\State $l \gets \argmax (T(L))$
\State $n \gets n+2$ 
\EndWhile
\State isRAok = ($\max(T) == 0$)
\end{algorithmic}
\end{algorithm}

\Cref{alg:input_partitioning} manages the information in a backward fashion, that is, it selects a s-zonotope with the higher time-to-last-error. This usually gives better results than working in a forward fashion (that is, selecting the s-zonotope with lower time-to-first-error) since it avoids to get stuck by exhaustively splitting up to the satisfaction of a constraint at time $t$, which, then, may have a small impact on the constraint satisfaction at $t+1$. On this subject, the algorithm can be straightforwardly adapted to handle the forward case by directly using \Cref{alg:reachability} (instead of \texttt{reach}), using $T(l) \gets N+1$ (instead of $T(l) \gets 0$) in Step 13  and selecting the minimum (instead of the maximum) of vector $T$.
Besides, the reduction operation used in \Cref{alg:reachability} must not truncate the initial symbols  even if their relevance decreases with time, so that the input-output mapping of the symbols identified by $I$ is preserved.


\subsection{Other possible settings and applications} \label{sec:discussion_agl}

Other choices for the proposed input partitioning strategy are as follows:
 \begin{itemize}  
    \item The strategy in \Cref{alg:input_partitioning} can  be adapted to handle open-loop verification problems (e.g. elements like the NN in isolation). In this case, the \texttt{reach} routine will only compute the output s-zonotope for the isolated element for a number of forward steps $N=1$.
    \item The maximum number of splits stopping criterion in \Cref{alg:input_partitioning} can be modified/complemented with a tolerance on the accuracy assessment  (see \Cref{sec:related_notions}). In other words, if the accuracy tolerance is fulfilled and a property is still violated, then the algorithm should stop to prevent from further splitting and the safety property is considered as unsatisfied up to the accuracy tolerance.
    \item  Another interesting application is to modify 
    the s-zonotope  split decision rule (Step 14 of \Cref{alg:input_partitioning}) to focus the split  in those regions for which the accuracy of using an affine abstraction is low (i.e. high $\|G_f\|_F$). This tends to return a set of initial s-zonotopes such that each 
    locally provides an accurate (affine) abstraction of the system behavior.   
\end{itemize}


\section{Simulations} \label{sec:simulations}


\subsection{Benchmark description}

The numerical simulations consist in  the discrete-time version of some of the verification problems proposed in the ARCH-COMP 2021~\cite{johnson2021arch}. Five dynamical systems are assessed, namely: single pendulum (\textbf{S}), TORA (\textbf{T)}, unicycle car (\textbf{C)}, adaptive cruise control (\textbf{ACC}) and double pendulum (\textbf{D}).  The above systems have been discretized using the forward Euler method with sampling period $\Delta T$, and they are controlled by a NN controller with control period $\Delta h$. The NN controllers are the ones provided in \cite{johnson2021arch} to control the continuous-time version of the models. To address this issue, the dynamical models have been analyzed under sampling times $\Delta T (\leq \Delta h)$ chosen sufficiently small for the discretization to have negligible impact in the model responses. Under this context, the same safety constraints and initial conditions than the ones proposed in~\cite{johnson2021arch} have been re-used to setup the reported simulations. Note that, as discussed in \Cref{rem:discrete_time}, the use of symbolic approaches supports the variation of $\Delta T$ (for some $\Delta h$) without inducing conservatism due to a loss  of dependencies between repeated control inputs. A detailed description of the systems dynamics can be found in \cite{johnson2021arch}, whereas the main parameters that define each safety verification problem are shown in \Cref{tab:param_benchmark}.


\begin{table}[h]
\centering
\begin{tabular}{cccccc}
\textbf{Problem} & \textbf{Dynamics} & \textbf{NN size}     & \textbf{H}   & $\Delta h$ &  $\Delta T$ \\ \hline
      \textbf{S1} & Single pendulum& (2,25,25,1) & 1 & .05 & .05   \\
      \textbf{S2} & Single pendulum& (2,25,25,1) & 1 & .05 & .001   \\
      \textbf{T1} & TORA & (4,100,100,100,1) & 20& 1& 1 \\ 
      \textbf{T2} & TORA & (4,100,100,100,1) & 20& 1& .01 \\ 
      \textbf{T3} & TORA & (4,100,100,100,1) & 20& 1& .001 \\ 
      \textbf{C1} & Unicycle car & (4,500,2) & 10& .2 & .2 \\
     \textbf{C2} & Unicycle car & (4,500,2) & 10& .2 & .001 \\
      \textbf{ACC1} & Cruise control & (5,20,20,20,20,20,1)& 5 & .1 & .1\\
      \textbf{ACC2} & Cruise control & (5,20,20,20,20,20,1)& 5 & .1 & .001\\
      \textbf{D1} & Double pendulum & (4,25,25,2) & 1 & .05 & .05\\
      \textbf{D2} & Double pendulum & (4,25,25,2) & 1 & .05 & .001\\
      \textbf{D3} & Double pendulum & (4,25,25,2) & 1 & .05 & .05\\
\end{tabular}
\caption{Benchmarks parameters}
\label{tab:param_benchmark}
\end{table}

\renewcommand{\arraystretch}{1.15}
\begin{table*}[t]
\centering
\begin{tabular}{ccccc}
\textbf{Problem} & \textbf{Time} [s] & \textbf{Initial set} & \textbf{Safety} &\textbf{Sft. Horizon} [s]  \\ \hline
\textbf{S1}  & 0.071  & {$[1,1.2] \times [0, 0.2]$} & {$x_1 \in [0,1]$} &  $0.5<t\leq 1$  \\
\textbf{S2} & 0.111 & {$[1,1.2] \times [0, 0.2]$} &  {$x_1 \in [0,1]$} &  $0.5<t\leq 1$ \\ 
\textbf{T1} & 0.036 & $[0.6,0.7]\times[-0.7,-0.6]\times[-0.4, -0.3]\times[0.5,0.6]$ & {$x \in [-2,2]^4$} & $t\leq 20$  \\
\textbf{T2} & 0.182 & $[0.6,0.7]\times[-0.7,-0.6]\times[-0.4, -0.3]\times[0.5,0.6]$  & {$x \in [-2,2]^4$} & $t\leq 20$ \\
\textbf{T3} & 1.515 & $[0.6,0.7]\times[-0.7,-0.6]\times[-0.4, -0.3]\times[0.5,0.6]$  &  {$x \in [-2,2]^4$}  & $t\leq 20$ \\ 
\textbf{C1}  & 0.37 & {$[9.5,9.55]\times[-4.5,-4.45]\times[2.1, 2.11]\times[1.5,1.51]$} & {$x \in [-.6,.6]\times [-.2, .2] \times [-.06,.06] \times [-.3,.3]$} & $t=10$   \\
\textbf{C2} & 21.40 & {$[9.5,9.55]\times[-4.5,-4.45]\times[2.1, 2.11]\times[1.5,1.51]$} & {$x \in [-.6,.6]\times [-.2, .2] \times [-.06,.06] \times [-.3,.3]$} & $t=10$ \\ 
\textbf{ACC1} & 0.095 & {$[90,110]\times[32,32.2] \times 0 \times [10 , 11] \times [30, 30.2] \times 0$} & {$x_1-x_4 \geq 10 + 1.4x_5$} & $t \leq 5$ \\
\textbf{ACC2} & 0.439 & {$[90,110]\times[32,32.2] \times 0 \times [10 , 11] \times [30, 30.2] \times 0$} & {$x_1-x_4 \geq 10 + 1.4x_5$} & $t \leq 5$  \\ 
\textbf{D1}  & 0.251 & $[1, 1.1]\times [1, 1.1]\times [1, 1.2]\times [1, 1.2] $ & {$x \in [-1,1.7]^4$}  & $t\leq 1$ \\ 
\textbf{D2} & 2.113 & {$[1, 1.1]\times [1, 1.1]\times [1, 1.2]\times [1, 1.2] $} & {$x \in [-1,1.7]^4$}  & $t\leq 1$  \\
\textbf{D3} & Fail & {$[1, 1.3]\times [1, 1.3]\times [1, 1.3]\times [1, 1.3] $} & {$(x_1,x_2) \in [-2,2]^2, \ (x_3,x_4) \in [-1.7,1.7]^2$} &  $t\leq 1$\\
& & & &
\end{tabular}
\caption{Benchmark problems and results for an s-zonotope implementation}
\label{tab:results_benchmark}
\end{table*}


\subsection{Benchmark results using s-zonotopes}

All the results reported below were obtained on a standard laptop with  Intel Core i7-8550U@1.8GHz$\times$4 processor and 16GB RAM running Windows 10. \Cref{tab:results_benchmark} shows the set of initial states, the safety constraints (with their time horizon), as well as the time required by an s-zonotope implementation to verify each problem. The reduction order is $q = 200$ in all the experiments. Some particularities are discussed below:
\begin{itemize}
    \item \textit{Single pendulum} (\textbf{S}): in a discrete-time setting, the constraint $x_1 \in [0, \, 1]$ is guaranteed to be satisfied for  problem \textbf{S1} (with $\Delta T = 0.05$s) for the time interval $t \in [0.55, \, 1]$  (that is, for samples $\{11,...,20\}$), whereas in \textbf{S2} (with $\Delta T = 0.001$s) the constraint satisfaction is guaranteed for the time interval $t \in [0.516, \, 1]$. 
    \item \textit{TORA} (\textbf{T}): in $\textbf{T1}$, the closed-loop system is not stable for the discrete-time model obtained for $\Delta T = 1$s. In this case, an unambiguous constraint violation is achieved at  $t = 3$ in $0.036$s. On the other hand, the closed-loop model obtained in \textbf{T2} and \textbf{T3} is stable, and the s-zonotope method verifies the satisfaction of the safety constraint in both problems without resorting to split the input set.  
    \item \textit{Unicycle car} (\textbf{C}): the model under study considers the addition of an unknown-but-bounded disturbance $w \in 10^{-4}[-1,\, +1]$ affecting the fourth state. The safety properties are verified for both \textbf{C1} and \textbf{C2}. In particular, \Cref{fig:car} shows the envelope computed for \textbf{C2} in the time interval $t \in [0, \, 10]$ and how the outer-approximation lies within the goal set at $t=10$. 
    \item \textit{Adaptative cruise control} (\textbf{ACC)}: both problems \textbf{ACC1} and \textbf{ACC2} are verified for the given time horizon.
    \item \textit{Double pendulum} (\textbf{D}): the set of constraints in problems $\textbf{D1-2}$ are violated by the closed-loop system. An unambiguous constraint violation is achieved for \textbf{D1} at $t = 0.25$ and for \textbf{D2} at $t=0.278$. 
    On the other hand, the problem \textbf{D3} cannot be verified from a simple affine abstraction: 
    the accumulated error indeed increases in the reachability analysis of \textbf{D3}, not allowing to guarantee the constraints satisfaction or their unambiguous violation, and thus motivating further extensions.   
\end{itemize}

The results presented above show how, despite their low computational complexity, s-zonotopes yield a high performance in NNCSs verification, being able to verify almost all the benchmark problems without splitting the input set. It is also remarkable the scalability of this approach. As an example, for problem \textbf{T3} with $\Delta T = 0.001$s, $\Delta h = 1$s and time horizon $t \in [0, 20]$, the proposed tool only  requires of $1.515$s to compute and assess $N=20\text{s}/\Delta T = 20.000$ forward iterations.

\begin{figure}[h]
\centering
  \subfloat[States $x_1$ vs $x_2$\label{fig:tora_12}]{%
       \includegraphics[width=\linewidth]{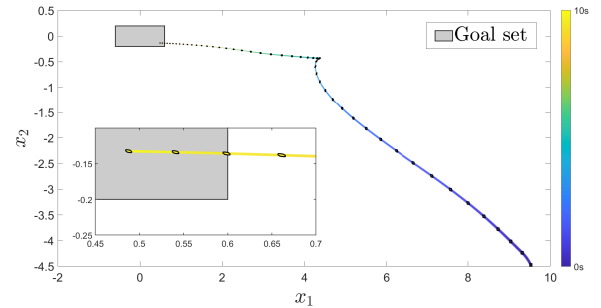}} \\
  \subfloat[States $x_3$ vs $x_4$\label{fig:tora_34}]{%
       \includegraphics[width=\linewidth]{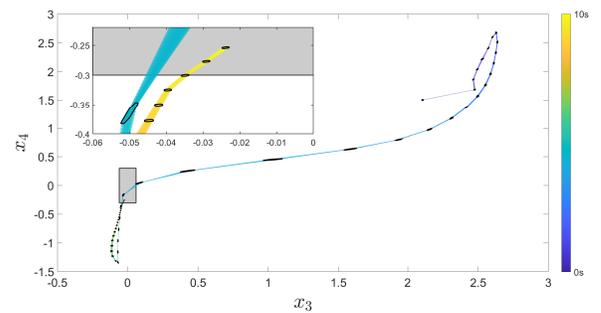}}
\caption{Problem \textbf{C2}: framed zonotopes represent the computed bounds at each $\Delta h = .2$s; blurred lines represent the bounds update at $\Delta T = .001$s.}
\label{fig:car}
\end{figure}


\subsection{Use of s-polynotopes} \label{sec:results_pol}

The capability of s-polynotopes to capture the non-convex map of NNs is illustrated below. To that end, the set of randomly generated neural networks used in  \cite{fazlyab2020safety} are analyzed. All the NNs consists of 2 inputs, 2 outputs and they differ on the number of hidden layers and neurons per layer. The first four NNs present $l = \{1,2,3,4\}$ hidden layers, each having $n_k = 100$ ReLU neurons per layer. The examined NN input set is $\mathcal{X}_0 =  [0.9, \, 1.1]\times [0.9, \, 1.1]$. \Cref{fig:pol_open} shows the set-valued interpretation of the output bounding s-polynotopes obtained  by abstracting the activity functions of active neurons with second order polynomials i.e. with $n=2$ in \eqref{eq:n-pol}. The computation times are $\{0.178, 0.240,2.021,3.329\}s$ for the NNs with $l=1$ to $l=4$ hidden layers, respectively. Similarly, another set of NNs with $l = \{7,8,9,10\}$ hidden layers and $n_k = 10$ ReLU neurons per layer is evaluated for the same input set. \Cref{fig:pol_open_2} represents the interpretation of the resulting s-polynotopes that are computed in $\{0.2389,0.108,0.786,0.155\}$s, respectively. Those examples (\Cref{fig:pol_open} and \Cref{fig:pol_open_2}) taken from \cite{fazlyab2020safety}  illustrate the ability of s-polynotopes composition to accurately generate inclusion preserving polynomial I/O mappings of NNs. As a byproduct, an efficient implicit description of possibly non convex output sets is obtained.

\begin{figure}[h!]
    \centering
     \includegraphics[width=\linewidth]{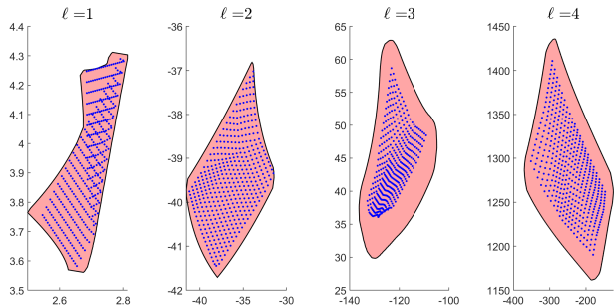} 
    \caption{NNs with 100 neurons per layer and $l = \{1,2,3,4\}$ hidden layers. Set-valued interpretation of the over-approximating s-polynotope (red set); exhaustive evaluation of the NNs (blue dots).}
    \label{fig:pol_open}
\end{figure}

\begin{figure}[h!]
    \centering
     \includegraphics[width=\linewidth]{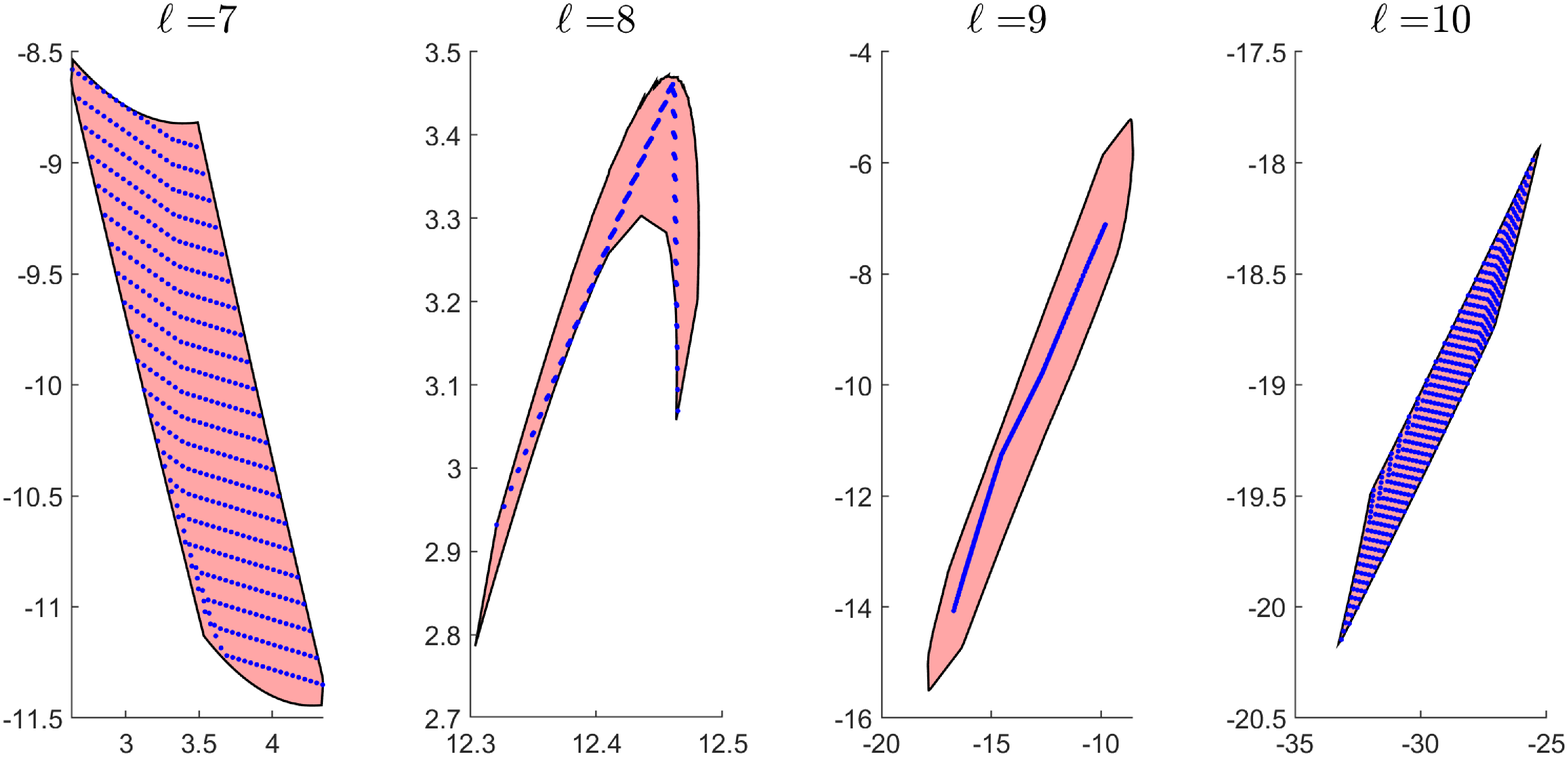} 
    \caption{NNs with 10 neurons per layer and $l = \{7,8,9,10\}$ hidden layers. Set-valued interpretation of the over-approximating s-polynotope (red set); exhaustive evaluation of the NNs (blue dots).}
    \label{fig:pol_open_2}
\end{figure}

\begin{figure}[t]
\centering
  \subfloat[Outputs\label{fig:robot_arm_split_out}]{%
       \includegraphics[width=\linewidth]{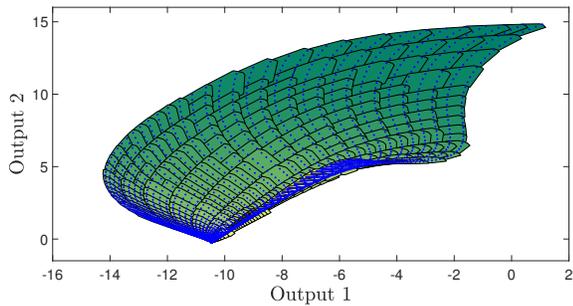}} \hfill
  \subfloat[Inputs \label{fig:robot_arm_split_inp}]{%
       \includegraphics[width=\linewidth]{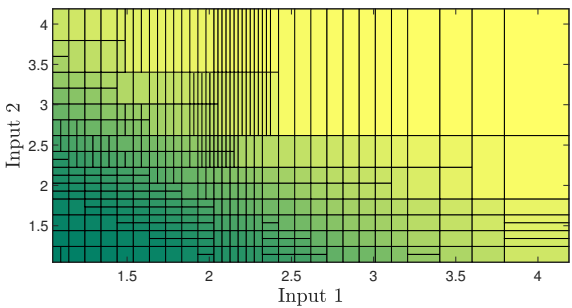}}
\caption{Robot arm example: input partitioning for $(\theta_1,\theta_2) \in [\frac{\pi}{3}, \, \frac{4\pi}{3}]^2$; yellow-green colors characterize the corresponding I/O set pairs; exhaustive evaluation of the NN (blue dots).}
\label{fig:robot_arm}
\end{figure}

\begin{table*}[t]
\centering
\begin{tabular}{c|c|c|cccccccc}
 \multirow{2}{*}{\textbf{d}}& \multirow{2}{*}{\textbf{Nb of splits}} & \multirow{2}{*}{\textbf{Time} [s]}   & 4 & 5 & 6 & 7 & 8 & 9 & 10 & \textbf{Nb of splits} \\  \cline{4-11} \\[-1em]
 & (Algo. \ref{alg:input_partitioning}) & &  (224) & (1344) & (8448) & ($5.49 \cdot 10^4$) & ($3.66 \cdot 10^5$) & ($2.49\cdot 10^6$) & ($1.71\cdot10^7$) & \textbf{Possible comb.}   \\ \hline
  1.75 & 5 & .052 & .446 & \textbf{.967} & $\star$ & $\star$ & $\star$ & $\star$ & $\star$ & \multirow{5}{*}{\thead{\textbf{$\bm{\%}$ of comb.} \\ \textbf{satisfying the} \\ \textbf{safety prop.}}} \\
  1.5 & 6 & .057 & -  & - & \textbf{.213} & $\star$ & $\star$ & $\star$ & $\star$ & \\
  1.3 & 7 & .072 & - & - & .095 & \textbf{.208}  & $\star$  & $\star$  &  $\star$ & \\
  1.2 & 8 & .058 & - & - & - & .032 & \textbf{.092} & $\star$ & $\star$ & \\
  1.1 & 10 & .066 & - & - & - & - & - & .002 & \textbf{.008} &
\end{tabular}
\caption{Partitioning performance:  non existing solutions (-);  worse \protect \footnotemark[2] solutions ($\star$).}
\label{tab:partitioning_performance}
\end{table*}


\subsection{Partitioning strategy}  \label{sec:results_partitioning}

Firstly, in order to show the performance of the partitioning algorithm, it will be applied to the open-loop robotic arm example used in \cite{xiang2020reachable,everett2020robustness}. Particularly, the non-linear dynamics of a 2 DOF robot arm are modeled by a $(2,5,2)$ NN with $tanh$ activations. The considered set of joint angles are extended to $(\theta_1,\theta_2) \in [\frac{\pi}{3}, \, \frac{4\pi}{3}]^2$. An  implementation of \Cref{alg:input_partitioning} adapted to analyze the NN in isolation is executed in order to iteratively minimize the $F$-radius\footnote{The $F$-radius of a zonotope is the Frobenius norm of its generator matrix (see Definition~3 in \cite{combastel2015zonotopes}).} of the zonotope spanned by the error symbols ($\|G_f\|_F$) for a fixed number of $n_{max} = 400$ splits. The computation time of the algorithm is $0.097$s. \Cref{fig:robot_arm_split_inp} shows the resulting pattern of $401$ input subsets, whereas \Cref{fig:robot_arm_split_out} represents the corresponding s-zonotope interpretations obtained in the output space altogether with an exhaustive evaluation of the NN (blue dots). This latter figure shows how \Cref{alg:input_partitioning} achieves an accurate description of the non-convex output set by focusing the splitting effort in those regions of the input space for which an affine abstraction granted by s-zonotopes is not accurate enough.

Furthermore, considering the initial set $(\theta_1,\theta_2) \in [\frac{\pi}{3}, \, \frac{2\pi}{3}]^2$, \Cref{alg:input_partitioning} is set to split up to the satisfaction of the safety constraint $y_1 \leq d$ (where $y_1$ denotes the first output).  \Cref{tab:partitioning_performance} reflects the number of splits and the time required by \Cref{alg:input_partitioning} to satisfy the above safety constraint for different values of $d$. Besides, \Cref{tab:partitioning_performance} also shows, for a fixed number of splits, the number of existing possible combinations of set selections and symbols bisections, as well as how many among them are able to satisfy the property. As an example, for $d = 1.2$, \Cref{alg:input_partitioning} requires $8$ splits. For the same problem, there are no possible combinations of less than 7 splits for which the property can be proven; there exist 18 out of $5.491\cdot 10^4$ possibilities that satisfy it with 7 splits ($0.032\%$); and 336 out of $3.66\cdot 10^5$ possibilities that satisfy it with 8 splits ($0.0918\%$). 

Regarding the closed-loop examination of \Cref{alg:input_partitioning}, this is applied to assess the satisfaction of problem \textbf{D3}, which cannot be satisfied by a simple s-zonotope abstraction. To that end, \Cref{alg:input_partitioning} is set to split up to the satisfaction of the safety constraints in \Cref{tab:results_benchmark}. The algorithm requires a total of 19 splits (i.e. 20 subsets) computed in $5.12$s. \Cref{fig:D1_split} shows the time evolution of the interval enclosure of the resulting 20 reachable sets (light blue background), altogether with 50 random simulations of the closed-loop system (blue dots).

\begin{figure}[h!]
    \centering
     \includegraphics[width=1\linewidth]{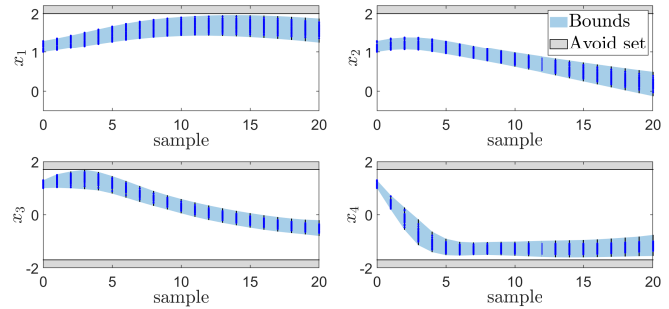} 
    \caption{Input partitioning problem \textbf{D3}: interval enclosure of the resulting 20 reachable sets (light blue); random simulations (blue dots).}
    \label{fig:D1_split}
\end{figure}

\footnotetext[2]{In the sense that require a higher number of splits to verify the property}


\section{Conclusions} \label{sec:conclusions}

A compositional approach focused on inclusion preserving long term symbolic dependency modeling is introduced in this work for the analysis of NNCSs, where such long term is to be understood both in time iterations (regarding the controlled system dynamics) and in layer iterations (regarding the sole NNs). This results in a generic method that has been developped in several ways. Firstly, the matrix structure of s-zonotopes enables to compute (fast and simple) affine symbolic mappings to abstract the I/O mapping of the control loop components. Two further extensions are also  proposed: the use of s-polynotopes to compute inclusion preserving polynomial mappings capable of accurately describing the non-convex map of NNs, and an input partitioning algorithm that benefits from the ability granted by s-zonotopes to preserve linear dependencies between the loop elements. Simulations show the comparative efficiency of the proposals and support the prevalence of dependency preserving methods for closed-loop analysis over the use of accurate, but dependency breaking, output bounding verification tools. Future works should address the integration with the analysis of continuous-time dynamical systems, as well as  the study of optimized affine/polynomial abstractions for achieving better performance in verifying specific safety properties.


\section*{Acknowledgments} 

The authors would like to thank Prof. Vicenç Puig (UPC, Barcelona, Spain) who initially prompted this collaboration and his support during the research stay of the first author at the University of Bordeaux.
This work was in part supported by the Margarita Salas grant from the Spanish Ministry of Universities funded by the European Union NextGenerationEU.


\appendices


\section{Proof of the \Cref{lem:general_inclusion}}\label{app:appendix_a} 
Define $\alpha = \frac{h(u)-h(l)}{u-l}$ and form $\xi(x) =  h(x) - \alpha x$ which is continuous in $[l, \,u]$, differentiable in $(l, \, u)$ and satisfies $\xi(l) = \xi(u) = 0$. Then, the maximum $\xi(\bar{x})$ (resp. minimum $\xi(\underline{x})$) of $\xi(x)$ on $[l, \, u]$ must be for $\bar{x}$ (resp. $\underline{x}$) in the boundary points $\{u,l\}$ or in its stationary points denoted as $\{\delta_1,...,\delta_n\}$, that is, the solutions of  $\xi^\prime(x) = 0 \to h^\prime(x) = \alpha$. 

Given $\xi(\bar{x})$ and $\xi(\underline{x})$, then for all $x \in [l,\,u]$  $\xi(\underline{x}) \leq \xi(x) \leq \xi(\bar{x}) \implies \underline{y}(x) \leq  h(x) \leq  \bar{y}(x)$, with  $\underline{y}(x) =   \alpha(x-\underline{x}) + h(\underline{x})$ and  $\bar{y}(x) =  \alpha(x-\bar{x}) + h(\bar{x})$. Thus, it follows that $\forall x \in [l, \, u], \exists \epsilon \in [-1,+1]$ such that
$$h(x) = \frac{\bar{y}(x)+\underline{y}(x)}{2} + \frac{\bar{y}(x)-\underline{y}(x)}{2}\epsilon$$  or equivalently $h(x) = \tilde{h}(x,\epsilon) =  \alpha x + \beta + \gamma \epsilon$ with
\begin{equation*}
    \begin{aligned}
    \beta &= \frac{\bar{y}(x)+\underline{y}(x)}{2} - \alpha x = \frac{h(\underline{x})+h(\bar{x})-\alpha(\underline{x}+\bar{x})}{2}, \\
    \gamma &= \frac{\bar{y}(x)-\underline{y}(x)}{2} = \frac{h(\bar{x})-h(\underline{x})+\alpha(\underline{x}-\bar{x}) }{2}.
    \end{aligned}
\end{equation*}


\section{Proof of the \Cref{prop:triplet_param_quad}} \label{app:appendix_b} 

\textit{Inclusion preservation:} the sign criterion $\gamma \geq 0$ is chosen below. Given $\bar{y}(x) = \alpha_2x^2+\alpha_1x+\beta+\gamma$ and $\underline{y}(x) = \alpha_2x^2+\alpha_1x+\beta-\gamma$, then parameters $(\alpha_2,\alpha_1,\beta,\gamma)$ ensure local coverage of $\varphi(x)$ if 
$\underline{y}(x) \leq \varphi(x) \leq \bar{y}(x), \forall x \in [l, \,u]$.

Consider firstly the scenario $|l| \leq u \leq 2|l|$. In this case, $\alpha_2 =  \frac{1}{2u} >0$ and thus $\underline{y}(x),\bar{y}(x)$ are strictly convex. 

On the one hand,  $\beta = \gamma$ and $\alpha_1 = 1-\alpha_2u$ impose that, $\underline{y}(0) = 0 = \varphi(0)$ and $\underline{y}(u) = u = \varphi(u)$, whereas  $\underline{y}(l) = \frac{1}{2}(\frac{l^2}{u}+l) \leq 0 = \varphi(l)$ for $u\geq |l|$ and $l<0$. Therefore, from $\underline{y}(l) \leq \varphi(l)$, $\underline{y}(0) = \varphi(0)$, $\underline{y}(u) = \varphi(u)$ and the convexity of $\underline{y}(x)$ wrt $x$, it follows that $\underline{y}(x) \leq \varphi(x) , \forall x \in [l, \, u]$.

On the other hand, from $\alpha_1 = 1-\alpha_2u$ and $\beta= \gamma = \frac{\alpha_2u^2}{8}$, then $\bar{y}(x)$ is tangent to the positive region of $\varphi(x)$ in $\hat{x} = \frac{u}{2}$ (that is, $\bar{y}(\hat{x}) = \varphi(\hat{x}) = \hat{x}$ and $\bar{y}^\prime(\hat{x}) = \varphi^\prime(\hat{x}) = 1$), and thus, since $\bar{y}(x)$ is convex, it follows that $\bar{y}(x) \geq \varphi(x), \forall x \geq 0$. Additionally,  for $\alpha_2 =  \frac{1}{2u}$ the (global) minimum of $\bar{y}(x)$ is  $\bar{y}(x^*) = 0$ for $x^* = \frac{-u}{2}$ (that is, $\bar{y}^\prime(x^*) = 0$), and thus $\bar{y}(x) \geq \varphi(x), \forall x \leq 0$.

A similar reasoning can be used to show the inclusion preservation for the scenario $u<|l|\leq 2u$.

\textit{Conservatism reduction:} For  the scenario $|l| \leq u \leq 2|l|$, the parameter $\gamma$ has the value $\gamma = \frac{\alpha_2u^2}{8} = \frac{u}{16}$. On the other hand, for a ReLU function $\varphi(x) = max(0,x)$ the triplet for an affine abstraction in \Cref{lem:triplets} yields $\gamma_{aff}^* = \frac{u|l|}{2(u+|l|)}$. Therefore, for $u \leq 2|l|$ the following inequality is obtained
\begin{equation*}
    \gamma_{aff}^* = \frac{u|l|}{2(u+|l|)} \geq \frac{u|l|}{2(2|l|+|l|)} = \frac{u}{6} > \frac{u}{16} = \gamma
\end{equation*}
and thus $\gamma \leq \frac{3}{8}\gamma_{aff}^* \sim |\gamma| \leq \frac{3}{8}|\gamma_{aff}^*|$ (since $\gamma,\gamma_{aff}^*>0$). A similar reasoning can be used to prove the case $u<|l|\leq 2u$.


\bibliographystyle{IEEEtran}
\bibliography{arxiv_nncs_verif.bib}

\begin{thebibliography}{10}
\providecommand{\url}[1]{#1}
\csname url@samestyle\endcsname
\providecommand{\newblock}{\relax}
\providecommand{\bibinfo}[2]{#2}
\providecommand{\BIBentrySTDinterwordspacing}{\spaceskip=0pt\relax}
\providecommand{\BIBentryALTinterwordstretchfactor}{4}
\providecommand{\BIBentryALTinterwordspacing}{\spaceskip=\fontdimen2\font plus
\BIBentryALTinterwordstretchfactor\fontdimen3\font minus
  \fontdimen4\font\relax}
\providecommand{\BIBforeignlanguage}[2]{{%
\expandafter\ifx\csname l@#1\endcsname\relax
\typeout{** WARNING: IEEEtran.bst: No hyphenation pattern has been}%
\typeout{** loaded for the language `#1'. Using the pattern for}%
\typeout{** the default language instead.}%
\else
\language=\csname l@#1\endcsname
\fi
#2}}
\providecommand{\BIBdecl}{\relax}
\BIBdecl

\bibitem{szegedy2013intriguing}
C.~Szegedy, W.~Zaremba, I.~Sutskever, J.~Bruna, D.~Erhan, I.~Goodfellow, and
  R.~Fergus, ``Intriguing properties of neural networks,'' \emph{arXiv preprint
  arXiv:1312.6199}, 2013.

\bibitem{kurakin2018adversarial}
A.~Kurakin, I.~J. Goodfellow, and S.~Bengio, ``Adversarial examples in the
  physical world,'' in \emph{Artificial intelligence safety and
  security}.\hskip 1em plus 0.5em minus 0.4em\relax Chapman and Hall/CRC, 2018,
  pp. 99--112.

\bibitem{liu2021algorithms}
C.~Liu, T.~Arnon, C.~Lazarus, C.~Strong, C.~Barrett, M.~J. Kochenderfer
  \emph{et~al.}, ``Algorithms for verifying deep neural networks,''
  \emph{Foundations and Trends{\textregistered} in Optimization}, vol.~4, no.
  3-4, pp. 244--404, 2021.

\bibitem{johnson2021arch}
T.~T. Johnson, D.~M. Lopez, L.~Benet, M.~Forets, S.~Guadalupe, C.~Schilling,
  R.~Ivanov, T.~J. Carpenter, J.~Weimer, and I.~Lee, ``Arch-comp21 category
  report: Artificial intelligence and neural network control systems (ainncs)
  for continuous and hybrid systems plants.'' in \emph{ARCH@ ADHS}, 2021, pp.
  90--119.

\bibitem{yang2019efficient}
G.~Yang, G.~Qian, P.~Lv, and H.~Li, ``Efficient verification of control systems
  with neural network controllers,'' in \emph{Proceedings of the 3rd
  International Conference on Vision, Image and Signal Processing}, 2019, pp.
  1--7.

\bibitem{tran2020nnv}
H.-D. Tran, X.~Yang, D.~Manzanas~Lopez, P.~Musau, L.~V. Nguyen, W.~Xiang,
  S.~Bak, and T.~T. Johnson, ``Nnv: the neural network verification tool for
  deep neural networks and learning-enabled cyber-physical systems,'' in
  \emph{International Conference on Computer Aided Verification}.\hskip 1em
  plus 0.5em minus 0.4em\relax Springer, 2020, pp. 3--17.

\bibitem{claviere2021safety}
A.~Clavi{\`e}re, E.~Asselin, C.~Garion, and C.~Pagetti, ``Safety verification
  of neural network controlled systems,'' in \emph{2021 51st Annual IEEE/IFIP
  International Conference on Dependable Systems and Networks Workshops
  (DSN-W)}.\hskip 1em plus 0.5em minus 0.4em\relax IEEE, 2021, pp. 47--54.

\bibitem{schilling2021verification}
C.~Schilling, M.~Forets, and S.~Guadalupe, ``Verification of neural-network
  control systems by integrating taylor models and zonotopes,'' \emph{arXiv
  preprint arXiv:2112.09197}, 2021.

\bibitem{ivanov2020verifying}
R.~Ivanov, T.~J. Carpenter, J.~Weimer, R.~Alur, G.~J. Pappas, and I.~Lee,
  ``Verifying the safety of autonomous systems with neural network
  controllers,'' \emph{ACM Transactions on Embedded Computing Systems (TECS)},
  vol.~20, no.~1, pp. 1--26, 2020.

\bibitem{sidrane2022overt}
C.~Sidrane, A.~Maleki, A.~Irfan, and M.~J. Kochenderfer, ``Overt: An algorithm
  for safety verification of neural network control policies for nonlinear
  systems,'' \emph{Journal of Machine Learning Research}, vol.~23, no. 117, pp.
  1--45, 2022.

\bibitem{everett2021efficient}
M.~Everett, G.~Habibi, and J.~P. How, ``Efficient reachability analysis of
  closed-loop systems with neural network controllers,'' in \emph{2021 IEEE
  International Conference on Robotics and Automation (ICRA)}.\hskip 1em plus
  0.5em minus 0.4em\relax IEEE, 2021, pp. 4384--4390.

\bibitem{xiang2020reachable}
W.~Xiang, H.-D. Tran, X.~Yang, and T.~T. Johnson, ``Reachable set estimation
  for neural network control systems: A simulation-guided approach,''
  \emph{IEEE Transactions on Neural Networks and Learning Systems}, vol.~32,
  no.~5, pp. 1821--1830, 2020.

\bibitem{everett2020robustness}
M.~Everett, G.~Habibi, and J.~P. How, ``Robustness analysis of neural networks
  via efficient partitioning with applications in control systems,'' \emph{IEEE
  Control Systems Letters}, vol.~5, no.~6, pp. 2114--2119, 2020.

\bibitem{dutta2019reachability}
S.~Dutta, X.~Chen, and S.~Sankaranarayanan, ``Reachability analysis for neural
  feedback systems using regressive polynomial rule inference,'' in
  \emph{Proceedings of the 22nd ACM International Conference on Hybrid Systems:
  Computation and Control}, 2019, pp. 157--168.

\bibitem{huang2019reachnn}
C.~Huang, J.~Fan, W.~Li, X.~Chen, and Q.~Zhu, ``Reachnn: Reachability analysis
  of neural-network controlled systems,'' \emph{ACM Transactions on Embedded
  Computing Systems (TECS)}, vol.~18, no.~5s, pp. 1--22, 2019.

\bibitem{hu2020reach}
H.~Hu, M.~Fazlyab, M.~Morari, and G.~J. Pappas, ``Reach-sdp: Reachability
  analysis of closed-loop systems with neural network controllers via
  semidefinite programming,'' in \emph{2020 59th IEEE Conference on Decision
  and Control (CDC)}.\hskip 1em plus 0.5em minus 0.4em\relax IEEE, 2020, pp.
  5929--5934.

\bibitem{fazlyab2020safety}
M.~Fazlyab, M.~Morari, and G.~J. Pappas, ``Safety verification and robustness
  analysis of neural networks via quadratic constraints and semidefinite
  programming,'' \emph{IEEE Transactions on Automatic Control}, 2020.

\bibitem{zhang2018efficient}
H.~Zhang, T.-W. Weng, P.-Y. Chen, C.-J. Hsieh, and L.~Daniel, ``Efficient
  neural network robustness certification with general activation functions,''
  \emph{Advances in neural information processing systems}, vol.~31, 2018.

\bibitem{tjeng2017evaluating}
V.~Tjeng, K.~Xiao, and R.~Tedrake, ``Evaluating robustness of neural networks
  with mixed integer programming,'' \emph{arXiv preprint arXiv:1711.07356},
  2017.

\bibitem{katz2017reluplex}
G.~Katz, C.~Barrett, D.~L. Dill, K.~Julian, and M.~J. Kochenderfer, ``Reluplex:
  An efficient smt solver for verifying deep neural networks,'' in
  \emph{International conference on computer aided verification}.\hskip 1em
  plus 0.5em minus 0.4em\relax Springer, 2017, pp. 97--117.

\bibitem{dutta2017output}
S.~Dutta, S.~Jha, S.~Sanakaranarayanan, and A.~Tiwari, ``Output range analysis
  for deep neural networks,'' \emph{arXiv preprint arXiv:1709.09130}, 2017.

\bibitem{althoff2015introduction}
M.~Althoff, ``An introduction to cora 2015,'' in \emph{Proc. of the workshop on
  applied verification for continuous and hybrid systems}, 2015, pp. 120--151.

\bibitem{tran2020verification}
H.-D. Tran, S.~Bak, W.~Xiang, and T.~T. Johnson, ``Verification of deep
  convolutional neural networks using imagestars,'' in \emph{International
  conference on computer aided verification}.\hskip 1em plus 0.5em minus
  0.4em\relax Springer, 2020, pp. 18--42.

\bibitem{singh2019abstract}
G.~Singh, T.~Gehr, M.~P{\"u}schel, and M.~Vechev, ``An abstract domain for
  certifying neural networks,'' \emph{Proceedings of the ACM on Programming
  Languages}, vol.~3, no. POPL, pp. 1--30, 2019.

\bibitem{wang2018formal}
S.~Wang, K.~Pei, J.~Whitehouse, J.~Yang, and S.~Jana, ``Formal security
  analysis of neural networks using symbolic intervals,'' in \emph{27th USENIX
  Security Symposium (USENIX Security 18)}, 2018, pp. 1599--1614.

\bibitem{xiang2018output}
W.~Xiang, H.-D. Tran, and T.~T. Johnson, ``Output reachable set estimation and
  verification for multilayer neural networks,'' \emph{IEEE transactions on
  neural networks and learning systems}, vol.~29, no.~11, pp. 5777--5783, 2018.

\bibitem{gowal2018effectiveness}
S.~Gowal, K.~Dvijotham, R.~Stanforth, R.~Bunel, C.~Qin, J.~Uesato,
  R.~Arandjelovic, T.~Mann, and P.~Kohli, ``On the effectiveness of interval
  bound propagation for training verifiably robust models,'' \emph{arXiv
  preprint arXiv:1810.12715}, 2018.

\bibitem{combastel2020functional}
C.~Combastel, ``Functional sets with typed symbols: Mixed zonotopes and
  polynotopes for hybrid nonlinear reachability and filtering,''
  \emph{Automatica}, vol. 143, 110457, 2022.

\bibitem{moore2009introduction}
R.~E. Moore, R.~B. Kearfott, and M.~J. Cloud, \emph{Introduction to interval
  analysis}.\hskip 1em plus 0.5em minus 0.4em\relax SIAM, 2009.

\bibitem{combastel2020distributed}
C.~Combastel and A.~Zolghadri, ``A distributed kalman filter with symbolic
  zonotopes and unique symbols provider for robust state estimation in cps,''
  \emph{International Journal of Control}, vol.~93, no.~11, pp. 2596--2612,
  2020.

\bibitem{singh2018fast}
G.~Singh, T.~Gehr, M.~Mirman, M.~P{\"u}schel, and M.~Vechev, ``Fast and
  effective robustness certification,'' \emph{Advances in neural information
  processing systems}, vol.~31, 2018.

\bibitem{combastel2005state}
C.~Combastel, ``A state bounding observer for uncertain non-linear
  continuous-time systems based on zonotopes,'' in \emph{Proceedings of the
  44th IEEE Conference on Decision and Control}.\hskip 1em plus 0.5em minus
  0.4em\relax IEEE, 2005, pp. 7228--7234.

\bibitem{jaulin2001interval}
L.~Jaulin, M.~Kieffer, O.~Didrit, and E.~Walter, ``Interval analysis,'' in
  \emph{Applied interval analysis}.\hskip 1em plus 0.5em minus 0.4em\relax
  Springer, 2001, pp. 11--43.

\bibitem{combastel2015zonotopes}
C.~Combastel, ``Zonotopes and kalman observers: Gain optimality under distinct
  uncertainty paradigms and robust convergence,'' \emph{Automatica}, vol.~55,
  pp. 265--273, 2015.

\bibitem{kolmanovsky1998theory}
I.~Kolmanovsky and E.~G. Gilbert, ``Theory and computation of disturbance
  invariant sets for discrete-time linear systems,'' \emph{Mathematical
  problems in engineering}, vol.~4, no.~4, pp. 317--367, 1998.

\end{thebibliography}

\rule{0mm}{5cm} \\

\begin{IEEEbiography}[{\includegraphics[width=1in,height=1.25in,clip,keepaspectratio]{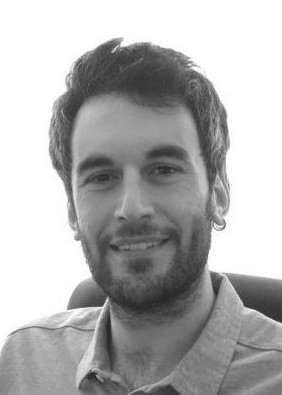}}]{Carlos Trapiello} received the M.Sc. degree in Aerospace Engineering from the Universidad Politécnica de Madrid (UPM), Madrid, Spain, in 2015; and the M.Sc. degree in Automatic Control \& Robotics and the Ph.D. degree in Control Engineering from the Universitat Politècnica de Catalunya-BarcelonaTech (UPC), Barcelona, Spain, in 2018 and 2021, respectively. From January 2022 to October 2022, he has been a postdoc fellow with the IMS lab, University of Bordeaux, France. He currently holds an industrial position in Toulouse, France.
\end{IEEEbiography}

\begin{IEEEbiography}[{\includegraphics[width=1in,height=1.25in,clip,keepaspectratio]{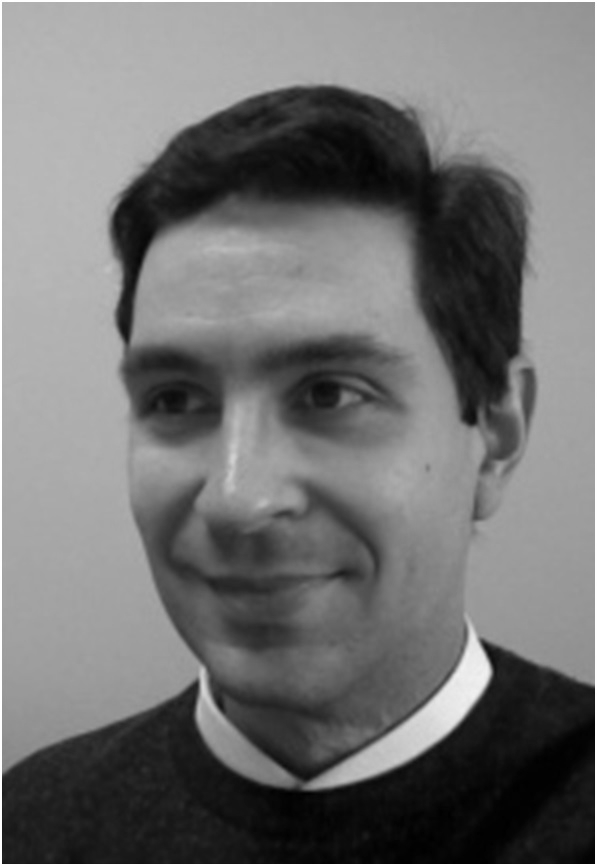}}]{Christophe Combastel} received his MSc degree in electrical engineering (1997) and his PhD in control systems (2000), both from the National Polytechnic Institute of Grenoble (Grenoble-INP), France. From 2001 to 2015, he was associate professor at ENSEA near Paris. Since 2015, he is with the University of Bordeaux and the IMS Lab (CNRS UMR5218). As a member of the ARIA team in the Control System Group of IMS, his research interests include interval, set-membership and stochastic algorithms for integrity control applications (e.g. aerospace) ranging from on-line diagnosis to verified model-based design, with special emphasis on uncertainty propagation, multi-sensor data fusion, and safety/security of cyber-physical systems.	
\end{IEEEbiography}

\begin{IEEEbiography}[{\includegraphics[width=1in,height=1.25in,clip,keepaspectratio]{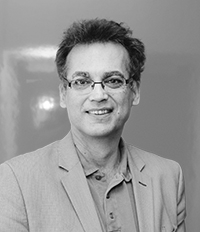}}]{Ali Zolghadri} received his PhD from the University of Bordeaux, France – and has been a Full Professor of Control Systems Engineering there since 2003. His current research interests are centered around autonomy, resilience and safety/(cyber)security of cyber-physical systems. He is member of International Technical Committees “SafeProcess” and “Aerospace” of IFAC – IEEE senior member, and TC member of EuroGNC (Council of European Aerospace societies). He has served as IPC member for various international conferences, and has delivered a number of plenary lectures and other invited talks at venues worldwide. He is an Associate Editor of the “Journal of the Franklin Institute” (Elsevier, USA) and “Complex Engineering Systems” journal – and Editorial Board member of “Aerospace Science and Engineering”, MDPI (Switzerland). He is author / co-author of more than 250 publications in archive journals, refereed conference proceedings and technical book chapters, and co-holder of 15 patents in aerospace. He is the recipient of CNRS Medal of Innovation 2016 which rewards – considering all fields and subfields of research – “\emph{outstanding scientific research with innovative applications in the technological and societal fields}”.
\end{IEEEbiography}

\vfill

\end{document}